\documentclass[amssymb, aps, pra, 10pt,superscriptaddress]{revtex4-2}
\usepackage{graphicx}
\usepackage{amsthm}
\usepackage{amsmath}
\usepackage{amssymb}
\usepackage{color}
\usepackage{soul}%To allow strike-out of text
\usepackage{url}
\usepackage{mathptmx}
\usepackage{times}
\usepackage{txfonts}
\usepackage{graphicx}% Include figure files
\usepackage{dcolumn}% Align table columns on decimal point
\usepackage{bm}% bold math
\usepackage{mathrsfs}
\usepackage{subfigure}
\usepackage{natbib}
\usepackage{ifsym}

\usepackage[usenames,dvipsnames,svgnames]{xcolor}

\begin{document}

\title{Nonlocal coherent states in an infinite array of boson sites}

\author{A P Sowa}
\affiliation{Department of Mathematics and Statistics, University of Saskatchewan, Canada}

\author{J Fransson}
\affiliation{Department of Physics and Astronomy, Uppsala University, Sweden}

%\author{Artur Sowa\\
%Department of Mathematics and Statistics, University of Saskatchewan\\
%106 Wiggins Road,
%Saskatoon, SK S7N 5E6,
%Canada \\
%sowa@math.usask.ca}
%\date{today}

\newtheorem{definition}{Definition}%[section]
\newtheorem{theorem}{Theorem}%[section]
\newtheorem{proposition}{Proposition}%[section]
\newtheorem{lemma}{Lemma}%[section]
\newtheorem{corollary}{Corollary}%[section]
\newtheorem{algorithm}{Algorithm}%[section]
\newtheorem{conjecture}{Conjecture}%[section]

\begin{abstract}

A regular coherent state (CS) is a special type of quantum state for boson particles placed in a single site. The defining feature of the CS is that it is an eigenmode of the annihilation operator. The construction easily generalizes to the case of a finite number of sites. %Indeed, it suffices to take the respective symmetric tensor product of single-site states.
However, the challenge is altogether different when one considers an infinite array of sites. %In this case, the direct application of the symmetric tensor power is not sufficient.
In this work we demonstrate a mathematically rigorous construction that resolves the latter case. The resulting nonlocal coherent states (NCS) are simultaneous eigenmodes for all of the infinitely many annihilation operators acting in the infinite array's Fock space. Our construction fundamentally relies on Dirichlet series-based analysis and number theoretic arguments.
%The Dirichlet series are arguably more complex and the questions related to them are the domain of analytic number theory, rather than the classical complex analysis. In particular, the resulting structure is perhaps less familiar than the ennobled theory for the regular coherent state. Nevertheless, as we will demonstrate, it is still possible to develop a framework that aids physical inquiry.
\vspace{.3cm}

\end{abstract}

\maketitle

  %\noindent AMS classification:

  \section{Introduction} The topic of coherent states (CS) has a long and distinguished tradition.  The 1985 compendium, \cite{Klauder}, features a historical introduction as well as 69 reprints, but the summary states that already at that point there were over 1000 publications on the subject. Contemporarily, the `mathscinet' displays a steady stream of mathematics publications involving CS in over half a century, with the total of over 1500 papers. A search on the `arXiv' reveals over 5600 papers that involve the CS or the related squeezed states. The original construction of the CS corresponds to the Heisenberg algebra.  Subsequently, the construction has been extended so that various Lie algebras result in their own genre of CS, e.g. the spin algebras, \cite{Klauder}, or $SL(2,\mathbb{R})$, \cite{Gazeau}, etc., including the case of infinite-dimensional Lie algebras, such as the loop and Kac-Moody algebras, \cite{Antonsen}. However, to our best knowledge, none of the preexisting work refers in any way to the Dirichlet series which, in a way, are the centerpiece of this work.

In this article, we are interested in a generalization of the classical CS to the case of a quantum system consisting of boson particles placed in an infinite array of sites. Such a system is somewhat akin to the well-known case of the ideal boson gas, consisting of an arbitrary number of indistinguishable bosonic particles, \cite{Merkli}. However, since boson sites in an array are distinguishable, a description of such a system requires substantially different mathematical structures. We emphasize that various constructions corresponding, either explicitly or implicitly, to a finite array of sites have been explored in the past. They play an important role in semi-classical analysis, see, e.g., \cite{Combescure-Robert}. However, in all those cases it suffices to consider finite tensor products of the classical coherent states.

  Our strategy to tackle the case of an infinite array of boson sites relies on a known device, first explored in \cite{Spector} and \cite{Bost_Connes}. Namely, we label the sites by primes rather than by consecutive integers. This has the effect of revealing the structure of the Fock space as that of $\ell_2(\mathbb{N})$ with a major role to play by the multiplicative number theory. In particular, some techniques of the analytic number theory become indispensable and, indeed, enable the construction of a rigorous theory. This framework also emphasizes the role of the multiplicative group of positive rationals $\mathbb{Q}_+$, which was investigated in a related context in \cite{Sowa_Fransson}. In Section \ref{section_construction} we construct a type of coherent states that we deem nonlocal, as indeed these states are distributed over the entirety of the infinite array of sites. The nonlocal coherent states (NCS) immediately bring out the role of the Dirichlet series. In particular, we discover a representation of the Fock space that is imprinted on a set of the Dirichlet series, see Subsection \ref{subsection_Dirichlet}. This representation may be viewed as a many-body analog of the Bargmann-Segal representation of the single-site Hilbert space. The main technical result is discussed in Section \ref{section_Cs}, where it is demonstrated that NCS arise from the vacuum state via an application of a suitably defined displacement operator, an effect analogous to that known in the context of classical CS. In Section \ref{section_Hamilt} we demonstrate how the NCS lead to explicit formulas for the expectation of certain types of Hamiltonians. 
  
  Section \ref{section_FT} highlights the role of the Fourier transform, and the Pontryagin duality, specialized to the group of positive rationals $\mathbb{Q}_+$. As an upshot, we obtain a generalized form of the NCS which differs by certain phase factors, parametrized by elements of the dual group of $\mathbb{Q}_+$, i.e., by points in the Cartesian product of infinitely many copies of $U(1)$. This also leads to a characterization of the Fourier transform on $\mathbb{Q}_+$ via the underlying quantum structure. It also sets the stage to introduce the ultimate form of the NCS, which enables the resolution of the identity, formula (\ref{res_identity_z}). This in turn enables a full resolution of the eigenvalues of the interaction Hamiltonian for a finite array. 
  In the closing sections, we present a sample of physical explorations enabled by the apparatus developed in previous sections. Since these inferences pertain to finite arrays of sites, they might have been supported by a simpler mathematical approach. Yet, basing them on a more general foundation, as we have done here, opens a broader view. It invites investigations into the nature of localization; in particular, whether delocalized states can play a meaningful role in a bosonic system's dynamics.

%\newpage
\section{Coherent states supported on an infinite array of boson sites}\label{section_construction}

We label the consecutive boson sites by primes $p = 2, 3, 5, \ldots$. %It is seen in what follows that this will have the effect recognizing the structure of the Fock space in $\ell_2(\mathbb{N})$.
Each site is associated with a bosonic annihilation operator $\hat{a}_{p}$ and the corresponding creation operator $ \hat{a}_{p}^\dagger$, which act as unbounded operators in the bosonic Fock space $\ell_2(\mathbb{N}) = \mbox{ span } \{ |k\rangle : k = 1,2,3, \ldots\}$. They satisfy the bosonic Canonical Commutation Relations, i.e.
\begin{equation}\label{BCCR}
[\,  \hat{a}_p, \hat{a}_q^\dagger \, ] = \delta_{p,q}, \quad  [\,  \hat{a}_p, \hat{a}_q \, ] = 0, \quad  [\,  \hat{a}_p^\dagger, \hat{a}_q^\dagger \, ] = 0\quad \mbox{ for all primes } p,q.
\end{equation}
We give an explicit definition of these operators. First, every integer is uniquely represented via
\[
k = \prod\limits_p p^{a_p(k)} \quad (\mbox{ the product is over distinct primes }).
\]
This defines the multiplicities $a_p(k)$. We set:
\begin{equation}\label{as}
    \hat{a}_p \, | k \rangle = \sqrt{a_p(k)}\,\, |k/p\rangle,  \quad
     \hat{a}_p^\dagger \, |k\rangle  = \sqrt{a_p(k)+1}\,\, | kp\rangle.
\end{equation}
Note that, in particular, if $p$ does not divide $k$, then $\hat{a}_p \, | k \rangle  =0$. Identity (\ref{BCCR}) is then verified by a straightforward calculation.
It is also convenient to define operators $\hat{a}_n$ for all $n\in \mathbb{N}$. Namely, we set:
\begin{equation}\label{aks}
 \hat{a}_1 = 1,\mbox{ and }\, \hat{a}_n = \prod\limits_p (\hat{a}_p)^{a_p(n)} \, \mbox{ for all }  n >1.
\end{equation}
Since all $\hat{a}_p$ commute, $\hat{a}_n$ are well defined. The following identity is an immediate consequence of the definition
\begin{equation}\label{akan}
 \hat{a}_m \hat{a}_n = \hat{a}_{mn} \, \mbox{ for all }  m,n.
\end{equation}
Next, we define a certain sequence, which will play central role in many calculations, namely:
\begin{equation}\label{def_xk}
  x_k := \sqrt{\prod\limits_p a_p(k)! }\quad k \in \mathbb{N}.
\end{equation}
%Note that $x_k$ tends to be large for those $k$ that are divisible by a high power of a prime. If, on the other hand, $k$ is a product of different primes, no matter how many, then $x_k = 1$.
Note that $k\mapsto x_k$ is a multiplicative function, i.e., $x_{m\cdot n} = x_m \cdot x_n $ whenever $m$ and $n$ are relatively prime.
%It is also easily seen that $x_k = (\exp^* \chi_{P}(k))^{-1/2}$, where $\chi_P$ is the indicator function of the set of primes, and $\exp^*$ is the exponential function for the Dirichlet convolution, \cite{Bateman-Diamond}.

The following fact is easily verified by a direct calculation:
\begin{equation}\label{an_onk}
  \hat{a}_n \, |k\rangle = \frac{x_k}{x_{k/n}} \,\, | k/n\rangle,
\end{equation}
with the convention that $| k/n\rangle = 0$ if $k$ is not divisible by $n$.
We also note that
\begin{equation}\label{an_hat_onk}
  \hat{a}_n^\dagger \, |k\rangle = \frac{x_{k n}}{x_{k}} \,\, | k\, n\rangle.
\end{equation}
Clearly, all operators $\hat{a}_n$ and $\hat{a}_n^\dagger$ are unbounded. Formulas (\ref{an_onk}) and (\ref{an_hat_onk}) imply that they are well defined on finite combinations of basis vectors, i.e. densely defined. It is worth emphasizing that the naive analogue of (\ref{BCCR}) does not hold true, i.e., the commutator $[  \hat{a}_m , \hat{a}_n^\dagger ]$ cannot be expected to equal $\delta_{m,n}$ in general.

We point out that the domains of operators $\hat{a}_p$ (resp. $\hat{a}_p^\dagger$), denoted $D(\hat{a}_p)$ (resp. $D(\hat{a}_p^\dagger)$), vary with $p$. Namely,
\begin{equation}\label{domains}
  D(\hat{a}_p) = D(\hat{a}_p^\dagger) = \{ v= \sum_k v_k\, |k\rangle\in \ell_2(\mathbb{N}):\, \sum_{k} \, |v_k|^2 a_p(k) < \infty \}.
\end{equation}
Note that $D(\hat{a}_p)$ is dense in $\ell_2(\mathbb{N})$.
It is seen directly that $\hat{a}_p^\dagger$ is indeed the adjoint of $\hat{a}_p$, i.e.
\begin{equation}\label{adjoints}
  \langle \, w\, | \,\hat{a}_p v \, \rangle = \langle \,\hat{a}_p^\dagger w\, | \, v \, \rangle \quad \mbox{ for all }\,\, v, w\in  D(\hat{a}_p) = D(\hat{a}_p^\dagger).
\end{equation}
Since $D(\hat{a}_p^\dagger) $ is dense, the operator $\hat{a}_p$ is closeable. Symmetric statements are true about $\hat{a}_p^\dagger$.

One of the fundamental operators that will be frequently at use is the number operator. The number operator at site $p$ is defined as $\hat{N}_p = \hat{a}_p^\dagger \hat{a}_p$. One readily obtains
\begin{equation}\label{N_p}
 \hat{N}_p \, | k\rangle = a_p(k)\, | k\rangle .
\end{equation}
The total number operator is defined as
$
 \hat{N} = \sum_{p} \hat{N}_p,
$
so that
\begin{equation}\label{N_in_canonical}
    \hat{N}\, | k\rangle   = \Omega(k)\, | k\rangle , \quad \mbox{ where }\quad \Omega(k):=  \sum_{p} a_p(k).
\end{equation}
Note that $\Omega(k)$ is the classical big prime-omega function.

It is essential to introduce finite-particle spaces. Namely, the zero-particle space is defined via $\mathbb{H}^{\odot 0} = \mathbb{C}$. Furthermore, for any integer $k\geq1$ the $k$-particle space is defined via
\begin{equation}\label{k-space}
\mathbb{H}^{\odot k} = \mbox{ span }\{\, |n\rangle : \Omega(n) = k \} = \{f : \hat{N} [f] = k\, f \},
\end{equation}
This gives $\ell_2(\mathbb{N})$ its structure as the Fock space:
\begin{equation}\label{ell2_hierarchy}
  \ell_2(\mathbb{N}) \equiv \bigoplus\limits_{k=0}^\infty \mathbb{H}^{\odot k}.
\end{equation}
Indeed, this is the Fock space of an array of boson sites, wherein the sites have been labeled by primes. 

Furthermore, let $\Pi_n: \ell_2(\mathbb{N})\rightarrow \mathbb{H}^{\odot n}$ be the orthogonal projection onto $\mathbb{H}^{\odot n}$. The domain of the number operator is characterized as follows
\[
D(\hat{N}) = \{ v\in   \ell_2(\mathbb{N}) : \, \sum_{n} n^2 \|\Pi_n v\|^2 < \infty \}.
\]
Note that $D(\hat{N})$ is a dense linear subset of $ \ell_2(\mathbb{N})$.

\subsection{Some number-theoretic preliminaries}

Everywhere below, $s = \sigma + it$ denotes a complex parameter with real part $\sigma$ and imaginary part $t$.
Below, we will make frequent use of the prime zeta function, defined as
\begin{equation}\label{prime_zeta}
  P_1 (s) = \sum_{p}  p^{-s}, \quad  \sigma > 1,
\end{equation}
where the sum is over all primes. It is well known that the series defining $P_1$ converges for $\sigma >1$, and diverges for $\sigma \leq 1$. We will also make use of generalizations of this function
\begin{equation}\label{n_prime_zeta}
  P_n(s) = \sum_{k: \Omega(k) = n}  k^{-s}, \quad  \sigma > 1,
\end{equation}
A special role is played by $P_2(s)$. It is useful to observe that
\begin{equation}\label{P_2_to_P_1}
  P_2(s) = \frac{1}{2}\left( P_1(s)^2 + P_1(2s) \right).
\end{equation}
Also, a special role is played by the series $\sum_k \frac{k^{-2\sigma}}{x_k^2}$, where $x_k$ are as in (\ref{def_xk}). Note that it converges at least as fast as $\sum_k k^{-2\sigma}$, and the latter converges whenever $\sigma > 1/2$.
The sum admits a more fundamental characterization via the prime zeta function, namely
\begin{equation}\label{the_N}
  \sum_k \frac{k^{-2\sigma}}{x_k^2} = \prod_p (1+ p^{-2\sigma} + \frac{1}{2!} p^{-2\sigma} +\ldots ) = \prod_{p} \exp p^{-2\sigma} = e^{P_1(2\sigma)}.
\end{equation}
It follows that
\begin{equation}\label{the_Pn_prep}
  \sum_{k: \Omega(k) = n} \frac{k^{-2\sigma}}{x_k^2} = \frac{1}{n!} P_1(2\sigma)^n.
\end{equation}
These facts will underlie some of the analysis of the NCS conducted below.
\vspace{.5cm}

\noindent
\emph{Remark:} In the context of this work, $P_1(s)$ is only relevant for $\Re s > 1$, where the defining series converges. Some properties of $P_1(s)$ are discussed in \cite{Froberg}. The author also considers the Dirichlet series $(1-P_1(s))^{-1} = \sum_{k} A_k k^{-s}$. It is demonstrated that $A_k = \Omega(k)!/x_k^2$, where $x_k$ is as in (\ref{def_xk}). The article also contains interesting insights into the growth rate of $k \mapsto A_k$.

\subsection{ The first definition of the NCS}

For any $s$ with $\sigma > 1/2$, we consider a superposition of basis states, given by:
\begin{equation}\label{the_s}
\sum_k \frac{k^{-s}}{x_k} \, |k\rangle,\quad \mbox{ where } x_k \mbox{ are as in (\ref{def_xk}). }
\end{equation}
We define the  nonlocal coherent state $|s\rangle$ by normalizing the superposition, i.e., in light of (\ref{the_N}), we set
\begin{equation}\label{the_s_all}
  |s\rangle = e^{-P_1(2\sigma)/2}\, \sum_k \frac{k^{-s}}{x_k}\, |k\rangle.
\end{equation}
This ensures that $\| \, |s\rangle \,\| = \sqrt{\langle s | s\rangle} = 1$. Note that those states $|k\rangle$ which have a lot of weight at one site $p$ are relatively suppressed in the superposition since that results in a large $x_k$. Also, the larger $\sigma$ is, the more weight is put on the initial states $|k\rangle$, and \emph{a fortiori} on the initial sites. When $\sigma \rightarrow \infty$ all the weight is shifted to the vacuum state $|1\rangle$. It is natural to ask for the expectation of operators $\Pi_n$, i.e. the probability of finding an $n$-particle state in $|s\rangle$. Identity (\ref{the_Pn_prep}) implies
\begin{equation}\label{Poisson}
 \langle s \, | \Pi_n |\, s\rangle = e^{-P_1(2\sigma)}\,\frac{1}{n!} P_1(2\sigma)^n.
\end{equation}
This is the Poisson distribution over the variable $n$ with the parameter $P_1(2\sigma) = \langle s \, | \hat{N} |\, s\rangle$, see (\ref{s_particle_number_better}).

Identity (\ref{an_onk}) implies
\begin{equation}\label{calc_n}
\hat{a}_n \,  \sum \frac{k^{-s}}{x_k} \,|k\rangle
  =  \sum \frac{k^{-s}}{x_k} \frac{x_k}{x_{k/n}} \, |\, k/n\rangle
  =  n^{-s} \, \sum \left(\frac{k}{n}\right)^{-s}\frac{1}{x_{k/n}} \, |\, k/n\rangle .
\end{equation}
Therefore,
\begin{equation}\label{s_is_eigen4n}
  \hat{a}_n \, |s \rangle = n^{-s} \, |s\rangle.
\end{equation}
Thus, each nonlocal coherent state $|s\rangle$ is an eigenstate for all annihilation operators simultaneously.
This is the most fundamental property of the NCS and, indeed, justifies the terminology. Also, note that even though operators $\hat{a}_n$ are unbounded, $|\langle s\, |\,  \hat{a}_n \, |s \rangle| < n^{-1/2}< \infty$, so these operators in a way ``mimic" bounded operators when probed on NCS. This type of phenomenon is well known in the context of bases, see e.g. \cite{Halmos}, Chapter 6.
\vspace{.5cm}

\noindent
\emph{Remark 1.} The NCS are Euler-type products of the regular coherent states. Indeed, series (\ref{the_s}) is an infinite symmetric tensor product of the form
\begin{equation}\label{inf_tens_pr}
 \bigodot\limits_{p} \, \sum_{n=1}^{\infty} \, \frac{p^{-ns}}{\sqrt{n!}} \, \, |p^n\rangle .
\end{equation}
Recall, see e.g. \cite{Klauder}, that the regular CS, denoted $|\, z\rangle$ with $z\in \mathbb{C}$, are proportional to
$
\sum_{k=1}^{\infty} \, \frac{z^{k}}{\sqrt{k!}} \, \, | k\rangle,
$
where $| k\rangle$ are the eigenstates of the harmonic oscillator.
Thus, each of the terms in the product (\ref{inf_tens_pr}) is proportional to the classical CS at site $p$ with a suitable substitution $z \mapsto p^{-s}$ and $|n\rangle \mapsto |p^n \rangle$. 
  However, the maps $z \mapsto p^{-s}$ are not one-to-one, as $p^{-s} = p^{- s+ 2 \pi i/\log p}$. Note that the infinite symmetric tensor product can ultimately be replaced by an ordinary product in the generalized Fourier picture, see (\ref{the_s_factorized}).
%\vspace{.5cm}
\vspace{.5cm}

\noindent
\emph{Remark 2.} The NCS are not the only eigenstates for all the annihilation operators. In Section \ref{section_FT} we provide a far reaching generalization of property (\ref{s_is_eigen4n}), namely (\ref{an_on_Uf_s}). As a result, we have a family of eigenstates parametrized by points in the half-plane and, additionally, by points on the Pontryagin dual of the group of positive rationals, which is an infinite-dimensional torus. This also highlights the special role of the unitary operator (\ref{def_U_mu}).  In Section \ref{section_res_I} we generalize the construction even further, obtaining NCS parametrized by infinite sequences of nonzero complex numbers (satisfying a certain technical assumption.) In particular, this leads to the conclusion, see Corollary \ref{cor_spectrum_a_n}, that the spectrum of every operator $\hat{a}_n$ is the entire complex plane.

\subsection{ NCS are minimal uncertainty states }

We will demonstrate that the NCS are minimal uncertainty states in a certain sense. To this end, define the displacement and momentum operators at site $p$ via
\begin{equation}\label{pos_mom}
  \hat{X}_p = \frac{1}{\sqrt{2}} \, \left( \hat{a}_p + \hat{a}_p^\dagger\right), \quad \hat{P}_p = \frac{1}{\sqrt{2}i} \, \left( \hat{a}_p - \hat{a}_p^\dagger\right).
\end{equation}
It follows that
\begin{equation}\label{pos_mom_com}
  [\,  \hat{X}_p, \hat{P}_q \, ] = i\,\delta_{p,q},  \quad
   [\,  \hat{X}_p, \hat{X}_q \, ] = 0, \quad   [\,  \hat{P}_p, \hat{P}_q \, ] = 0,
\end{equation}
which implies the Heisenberg uncertainty principle $\sigma_{\hat{X}_p}\, \sigma_{\hat{P}_p} \geq 1/2$  relative to any state (with the choice of units that gives $\hslash = 1$). However, one readily verifies  that when the state is $| s\rangle$, then
\[\sigma_{\hat{X}_p}^2 = \langle s \,|\, \hat{X}_p^2 \,|\, s \rangle - \langle s \,|\, \hat{X}_p \,|\, s \rangle^2 = \frac{1}{2} =
\langle s \,|\, \hat{P}_p^2 \,|\, s \rangle - \langle s \,|\, \hat{P}_p \,|\, s \rangle^2 = \sigma_{\hat{P}_p}^2.
\]
This means that all $|s\rangle$ are minimal uncertainty states for all  $p$.

\subsection{The map from the half-plane to the Fock space}

Next, we examine the map $\{s:\, \Re s > 1/2\}\ni s \mapsto |s\rangle \in \mathbb{H}$. %Throughout the section we assume $s_1, s_2 \in \{s:\, \Re s > 1/2\}$.
A direct calculation similar to (\ref{the_N}) gives:
  \begin{equation}\label{prod}
  % \langle s_1\, |\, s_2 \rangle = e^{-P_1(2\sigma_1)/2}\, e^{-P_1(2\sigma_2)/2}\, e^{P_1(s_1^* + s_2)} .
     \langle s_1\, |\, s_2 \rangle = \exp \left(-P_1(2\sigma_1)/2 -P_1(2\sigma_2)/2 + P_1(s_1^* + s_2)\right) .
  \end{equation}
Since
  $
  \|\, |s_1\rangle - |s_2\rangle \,\|^2 = 2 -  2\Re \langle s_1\, |\, s_2 \rangle ,
  $
it follows that the function $s_1 \mapsto  \|\, |s_1\rangle - |s_2\rangle \,\|$ (with a fixed $s_2$) is smooth. In particular,
\begin{equation}\label{top}
   \|\, |s_1\rangle - |s_2\rangle \,\| = O(|s_1 -s_2|)\quad  \mbox{ as } s_1 \rightarrow s_2.
\end{equation}

\subsection{Overcompleteness}

Next, we demonstrate that the set of all NCS is very ample; indeed, a small part of it is already a dense subset of the Hilbert space. Recall, a subset of $\ell_2(\mathbb{N})$ is dense if its closure is all of $\ell_2(\mathbb{N})$, \cite{Kato}.
We have the following:
\begin{theorem}\label{S_dense}
Let  $S\subset \{s: \Re s > 1/2\}$ be an infinite set of points that has an accumulation point. Then the corresponding set of coherent states $\tilde{S} = \{ |s\rangle : s\in S\}$ is a dense subset of $\ell_2(\mathbb{N})$.
\end{theorem}
\begin{proof}
One needs to demonstrate that if $|v\rangle \in \ell_2(\mathbb{N})$ and $\langle v | s\rangle = 0$ for all $|s\rangle \in \tilde{S}$, then necessarily $|v\rangle = 0$. To this end, to a state $|v\rangle = \sum_{k} v_k \, |k\rangle$ we assign the Dirichlet series $\Phi_v(s)$ by setting
 \begin{equation}\label{v_2_Phi}
 \Phi_v(s) := e^{P_1(2\sigma)/2}\, \langle v | s\rangle = \sum_{k}\, \frac{ v_k^*}{x_k}\, k^{-s}.
 \end{equation}
Observe that the series converges uniformly for all $s$, such that $\Re s \geq \sigma_0 > 1/2$. Indeed, for $N>1$, we have
 \[
 \sum_{k = N}^{\infty}\, \left| \frac{ v_k^*}{x_k}\, k^{-s}\right| \leq
 \sqrt{ \sum_{k = N}^{\infty}\, |v_k|^2}
 \,\,\sqrt{ \sum_{k = N}^{\infty}\,  \frac{ k^{-2\sigma_0}}{x_k^2} } \leq \| v\| \, e^{P_1(2\sigma_0)/2},
 \]
 and both terms in the middle approach zero when $N\rightarrow \infty$. This implies that $\Phi(s)$ is an analytic function in the half plane $\Re s > 1/2$. Since, by assumption, $\Phi(s)=0$ for all $s\in S$, and $S$ has an accumulation point, $\Phi(s)$ must be identically zero in the whole half plane $\Re s > 1/2$. Furthermore, a Dirichlet series that vanishes identically must have all coefficients equal to zero, hence  $v_k=0$ for all $k$, i.e., $|v\rangle =0$. This completes the proof.
 \end{proof}

\noindent
\emph{Remark.} With the simplifying assumption that $S$ is a vertical line, say, $S = \{s: \, \Re s = \alpha> 1/2\}$, an alternative proof of completeness may be derived from the Perron's inversion formula. Indeed, for a non-integer $r>1$, we have
 \begin{equation}\label{Perron}
   \sum_{k< r} \, \frac{ v_k^*}{x_k} = \frac{1}{2\pi i} \, \lim\limits_{T \rightarrow \infty} \int_{\alpha-iT}^{\alpha +iT} \Phi_v(s)\, r^s \,\frac{ds}{s}.
 \end{equation}
  However, $ \Phi_v(s)$ vanishes by assumption. Thus, taking $r = 3/2, 5/2, 7/2, \ldots$, we conclude that all partial sums vanish, which implies $|v\rangle =0$.

 \subsection{The Fock space encoded via the Dirichlet series} \label{subsection_Dirichlet}

The mapping $v \mapsto \Phi_v(s)$ given by (\ref{v_2_Phi}) enables us to give an alternative representation of the Fock space (\ref{ell2_hierarchy}). Namely,
\begin{equation}\label{DirF}
  \bigoplus\limits_{k=0}^\infty \mathbb{H}^{\odot k} = \left\{ \sum_{k} d_k \, k^{-s}: \sum_{k} x_k^2\, |d_k|^2   < \infty\right\},\quad \mbox{ where the weights } x_k \mbox{ are as in (\ref{def_xk})}.
\end{equation}
The corresponding Hermitian product is given by
\[
\left\langle \, \sum_{k} c_k \, k^{-s}\, | \, \sum_{k} d_k \, k^{-s}\, \right\rangle = \sum_{k} x_k^2 \, c_k^* d_k .
\]
%Thus, the arithmetic function $k \mapsto x_k^2$ plays the role of a weight function. 
Furthermore, the creation and annihilation operators assume the form
\begin{eqnarray}
% \nonumber % Remove numbering (before each equation)
  \hat{a}_p^\dagger \, \sum_{k} d_k \, k^{-s}  &=& \sum_{k} d_k \, (kp)^{-s} \\
  \nonumber
   && \\
  \hat{a}_p \, \sum_{k} d_k \, k^{-s} &=&  \sum_{k} d_{kp}\, [a_p(k) + 1] \, k^{-s}
\end{eqnarray}
Indeed, these identities are easily established by calculating $\langle \, v\, | \hat{a}\, | s\rangle$, resp. $\langle \, v\, | \hat{a}^\dagger\, | s\rangle$, and interpreting these expressions via $\Phi_{\hat{a}_p^\dagger v}$, resp. $\Phi_{\hat{a}_p v}$. It is also interesting to observe that
\begin{equation}\label{Np_on_Phi}
  \hat{N}_p \, \sum_{k} d_k \, k^{-s} = \, \sum_{k} \, a_p(k) \, d_k \, k^{-s},
\end{equation}
so that
\begin{equation}\label{N_on_Phi}
  \hat{N} \, \sum_{k} d_k \, k^{-s} = \, \sum_{k} \, \Omega(k)\, d_k  \, k^{-s}.
\end{equation}

\section{ NCS are generated by a family of unitary transforms}\label{section_Cs}

We introduce a few specific operators which play an important role in understanding the properties of NCS. The first family of operators is defined as follows:
  \begin{equation}\label{def_Cs}
      C_s = \sum_{p}p^{-s} \hat{a}_p, \quad ( \sigma = \Re s > 1/2 ), \quad \mbox{ with the domain }  D(\hat{N}^{1/2}) = \{ v\in   \ell_2(\mathbb{N}) : \, \sum_{n} n\, \|\Pi_n v\|^2 < \infty \}.
  \end{equation}
  Note that the restriction of $C_s$ to $\mathbb{H}^{\odot 0}$ is the trivial null operator.
  We also define:
\begin{equation}\label{def_Cs_dag}
 C_s^\dagger = \sum_{p}p^{-s^*} \hat{a}_p^\dagger, \quad ( \sigma = \Re s > 1/2 ), \quad \mbox{ with the domain } D(\hat{N}^{1/2}).
\end{equation}
Note that $D(\hat{N}^{1/2})$ is a dense linear subset of $ \ell_2(\mathbb{N})$.
The following theorem highlights the basic properties of these operators.

\begin{theorem}
\label{theorem_C}
 Operators $C_s, C_s^\dagger$ have the following properties:
\begin{enumerate}

  \item
    Let $ C_s^\dagger(n)$  be the restriction of $C_s^\dagger$ to the subspace $\mathbb{H}^{\odot n}$. The range of $  C_s^\dagger(n)$ is a linear subset of  $\mathbb{H}^{\odot n+1}$, and
  \begin{equation}\label{Cdaghier}
    C_s^\dagger(n) : \mathbb{H}^{\odot n}\rightarrow \mathbb{H}^{\odot n+1}
  \end{equation}
 is a bounded operator. Moreover, we have
 \begin{equation}\label{C_norms}
\|\, C_s^\dagger(n)\,\| = \left\{ \begin{array}{rll}
                                     = & \sqrt{n+1}\, \sqrt{P_1(2\sigma)} &\mbox{ for } n = 0,1 \mbox{ and } \sigma =\Re s >1/2 \\
                                     &  &  \\
                                    \leq  &  \sqrt{n \, P_2(\sigma) +P_1(2\sigma)} < \sqrt{n+1}\sqrt{ P_2(\sigma)} & \mbox{ whenever } n\geq 2 \mbox{ and } \sigma =\Re s >1.
                                  \end{array}\right.
 \end{equation}
We do not establish whether the tighter estimate for $n\geq 2$ is sharp.

  \item
   Let $ C_s(n)$  be the restriction of $C_s$ to the subspace $\mathbb{H}^{\odot n+1}$.
   The range of $  C_s(n)$ is a linear subset of $\mathbb{H}^{\odot n}$, and
  \begin{equation}\label{Chier}
    C_s(n) : \mathbb{H}^{\odot n+1}\rightarrow \mathbb{H}^{\odot n}
  \end{equation}
  is the adjoint of $C_s^\dagger(n)$, i.e.
  \begin{equation}\label{adjoint}
     \langle \,  u \, | \, C_s^\dagger(n) \, v\,\rangle =
      \langle \, C_s(n) \, u \, |\, v\,\rangle
      \quad \mbox{ for all }\quad u \in \mathbb{H}^{\odot n+1}, v \in \mathbb{H}^{\odot n}.
  \end{equation}
 In particular, $  C_s(n)$ is a bounded operator with norm $ \|\, C_s(n)\,\| =\|\, C_s^\dagger(n)\,\|$ as in (\ref{C_norms}).

 \item
  Whenever, $\sigma > 1$, operators $ C_s$ and $C_s^\dagger$ are closable. In addition their sum  $ C_s^\dagger + C_s$ and difference $ C_s^\dagger - C_s$ are also closable operators.

\end{enumerate}

\end{theorem}

\begin{proof}
\begin{enumerate}
  \item
   First, note that
  \[
     C_s^\dagger\, |k\rangle = \sum_{p}  p^{-s^*} \, \sqrt{a_p(k) + 1} \, |kp\rangle ,
  \]
  and
    \[
  \|\, C_s^\dagger \, |k  \rangle \, \|^2 = \sum_{p} \left|\,  p^{-s^*} \, \sqrt{a_p(k) + 1} \,  \right|^2 \leq (\Omega(k) + 1) \, P_1(2\sigma).
  \]
  This shows that (\ref{Cdaghier}) holds, provided $C_s^\dagger$ is bounded on the $n$-particle subspaces.
 Next, let $ v \in \mathbb{H}^{\odot n}$, say, $v = \sum_{k: \Omega(k) = n} v_k\, |k\rangle$. We obtain
\begin{equation}\label{Av}
     C_s^\dagger\, v = \sum_{k: \Omega(k) = n} v_k\,\sum_{p}  p^{-s^*} \, \sqrt{a_p(k) + 1} \, |kp\rangle  \\
     = \sum_{l: \Omega(l) = n+1} \,\sum_{p|l}  v_{l/p} \, p^{-s^*} \, \sqrt{a_p(l)} \, \, |l\rangle,
\end{equation}
where $\sum_{p|l}$ denotes the sum over all primes $p$ that are divisors of $l$.
   Thus,
  \[
  \| \,  C_s^\dagger v \,\|^ 2 = \sum_{l: \Omega(l) = n+1} \,\left|\, \sum_{p|l}  v_{l/p} \, p^{-s^*} \, \sqrt{a_p(l)}\, \right|^2 .
  \]
     To estimate this norm, we proceed in steps, fist observing the outcome for $n=0$ and then $n=1$ before considering the case of $n\geq 2$. When $n=0$, all $l$ with $\Omega(l) = 1$ are primes, so that $a_p(l) \in\{0, 1\} $, and
  \[
   \| \,  C_s^\dagger v \,\|^ 2 = 
  \sum_{p}  |v_1|^2 \, p^{-2\sigma}  = |v_1|^2\, P_1(2\sigma) = \|v\|^2\, P_1(2\sigma).
  \]
 Thus,  $\|C_s^\dagger (0)\| = \sqrt{P_1(2\sigma)} $ as claimed. Next, let $n=1$ so that $l = 4,6,9,10, \ldots$ and $a_p(l)\in\{0,1,2\}$. We reorganize and evaluate $ \sum_{l: \Omega(l) = 2} \,\left(\, \sum_{p|l}  |v_{l/p}|\, \sqrt{a_p(l)} \, p^{-\sigma} \, \right)^2$ as follows:
 \[
 \begin{split}
  & (|v_2| \sqrt{2}\, 2^{-\sigma})^2 + (|v_3| 2^{-\sigma} + |v_2| 3^{-\sigma})^2 + (|v_3| \sqrt{2}\, 3^{-\sigma})^2
    + (|v_5| 2^{-\sigma} + |v_2| 5^{-\sigma})^2+ \ldots
    \\
    & %\mbox{ (regrouping) }
    \\
     = &\,\, |v_2|^2 (2^{-2\sigma} +3^{-2\sigma} +5^{-2\sigma} +\ldots )
      \\
   + &  \,\, |v_3|^2 (2^{-2\sigma} +3^{-2\sigma} +5^{-2\sigma} +\ldots )
      \\
    +  & \,\, \ldots
     \\
     + & \,\, |v_2|\, 2^{-\sigma} (|v_2|\, 2^{-\sigma} +|v_3|\, 3^{-\sigma} +|v_5|\, 5^{-\sigma} +\ldots )
     \\
     + & \,\, |v_3|\, 3^{-\sigma} (|v_2|\, 2^{-\sigma} +|v_3|\, 3^{-\sigma} +|v_5|\, 5^{-\sigma} +\ldots )
     \\
    + & \ldots
     \\
     =& \|v\|^2 \, P_1(2\sigma)\,
     +\, \left(|v_2|\, 2^{-\sigma} +|v_3|\, 3^{-\sigma} +|v_5|\, 5^{-\sigma} +\ldots \right)^2 \, \leq \, 2  \|v\|^2 \, P_1(2\sigma) .
 \end{split}
 \]
 Taking $v_p$ such that $|v_p| = p^{-\sigma}$ we see that the estimate cannot be improved.
Therefore,  $\|C_s^\dagger (1)\| = \sqrt{2} \sqrt{P_1(2\sigma)} $  as claimed.

The structure of the proof changes for $n\geq 2$, but it is essentially the same for all these cases. First note
  \[
  \sum_{l: \Omega(l) = n+1} \,\left(\, \sum_{p|l}  |v_{l/p}| \, \sqrt{a_p(l)}\, p^{-\sigma} \, \right)^2
  =   \sum_{l: \Omega(l) = n +1} \,\sum_{p|l}  |v_{l/p}|^2 \, a_p(l)\, p^{-2\sigma}\, + \,
  \sum_{l: \Omega(l) = n +1} \,\sum_{p\neq q; p,q|l}  |v_{l/p} v_{l/q}| \, \sqrt{a_p(l)a_q(l)} \, (pq)^{-\sigma} =: \alpha + \beta.
  \]
  We regroup and estimate the first sum ($\alpha$) as follows:
  \[
  \alpha = \sum_{p} \, p^{-2\sigma}\sum_{k: \Omega(k) = n} \, (a_p(k)+1)  |v_{k}|^2 \, \leq P_1(2\sigma) \, (n+1)\, \|v\|^2.
  \]
  As to the second sum ($\beta$), we first note that if $|v_{l/p} v_{l/q}| \, \sqrt{a_p(l)a_q(l)} \neq 0 $, then $l = p q s$ for an integer $s$ with $\Omega(s) = n-1$, and so $|v_{l/p} v_{l/q}| \, \sqrt{a_p(l)a_q(l)}  = |v_{qs} v_{ps}| \, \sqrt{(a_p(s)+1)(a_q(s)+1)} $. This allows us to rewrite this term in the form
  \[
 \beta =  \sum_{p\neq q} \, (pq)^{-\sigma}  \sum_{s: \Omega(l) = n-1}\, |v_{qs} v_{ps}| \, \sqrt{(a_p(s)+1)(a_q(s)+1)}.
  \]
  Note that $\sqrt{(a_p(s)+1)(a_q(s)+1)} \leq n/2$, whereas $\sum_{s}\, |v_{qs} v_{ps}| \leq \|v\|^2 $ for any pair $p,q$. Therefore,
  \[
  \beta \leq \frac{n}{2} \sum_{p \neq q} \, (pq)^{-\sigma}\, \|v\|^2= \frac{n}{2} \, \left(P_1(\sigma)^2 - P_1(2\sigma)\right)\, \|v\|^2.
  \]
  Note that for convergence it is required that $\sigma > 1$.
In summary, taking into account (\ref{P_2_to_P_1}), we have $\alpha + \beta \leq \left( n \, P_2(\sigma) +P_1(2\sigma) \right)\, \|v\|^2$, which means that for $\sigma > 1$ we have $\| C_s^\dagger(n) \| \leq \sqrt{n \, P_2(\sigma) +P_1(2\sigma)} < \sqrt{n+1}\sqrt{ P_2(\sigma)}$.
 This completes the proof of the first statement.

 \item
 We have
  \[
  C_s\, |k\rangle = \sum_{p|k}p^{-s} \, \sqrt{a_p(k)} \, |\,k/p\rangle,
  \]
which means that (\ref{Chier}) holds provided $C_s(n)$ is bounded. Next, it is seen by inspection that (\ref{adjoint}) holds. This implies that $C_s(n)$ is the adjoint of $C_s^\dagger(n)$ and $\|\, C_s(n)\,\| = \|\, C_s^\dagger(n)\,\|$.  This completes the proof of the second statement.

  \item
 Estimates (\ref{C_norms}) imply that whenever $\sigma >1$, $C_s$ and $C_s^\dagger$ are both well defined in $D(\hat{N}^{1/2})$, which is dense in $\ell_2(\mathbb{N})$. To demonstrate that $C_s^\dagger$ is closable, we need to show that if a sequence $\|v^{m}\| \rightarrow 0$ and $\|C_s^\dagger v^{m} - w\| \rightarrow 0$, then $w=0$, see e.g., Chapter III, \cite{Kato}. Recall that $\Pi_n: \ell_2(\mathbb{N})\rightarrow \mathbb{H}^{\odot n}$ is the orthogonal projection. We have $\|\Pi_n v^{m}\| \rightarrow 0$, and in light of (\ref{Cdaghier}) $\|C_s^\dagger \Pi_n v^{m} \| \rightarrow 0$ for all $n$. Since $C_s^\dagger \Pi_n v^{m} = \Pi_{n+1} C_s^\dagger v^{m}$, this implies $\Pi_{n+1} w = 0$ for all $n\geq 0$ and, of course, $\Pi_0 w = 0$. Therefore $w=0$. The proof that $C_s$ is closable is similar, and based on identity (\ref{Chier}). Also, if a sequence $\|v^{m}\| \rightarrow 0$ and $\|(C_s \pm C_s^\dagger)\, v^{m} - w\| \rightarrow 0$, then $w=0$. Indeed, it suffices to observe that $\Pi_n\,(C_s \pm C_s^\dagger)\, v^{m} = C_s\Pi_{n+1} \, v^m \pm C_s^\dagger \Pi_{n-1}\, v^{m} \rightarrow 0$  and hence $\Pi_n \, w=0$ for all $n$. Thus, operators $ C_s \pm C_s^\dagger$ are closable. This also follows directly from the general fact that a symmetric densely defined operator is closable, see e.g., Chapter V, \cite{Kato}.

\end{enumerate}
\end{proof}
\vspace{.2cm}

In light of Theorem \ref{theorem_C}, both $C_s$ and $C_s^\dagger$ are closeable. Henceforth we will understand  $C_s$ and $C_s^\dagger$ to denote the maximal closed extensions. The theorem also implies that
\[
D(C_s) = \{u: \, \sum\limits_{k=0}^\infty \| C_s \, \Pi_{k+1}\, u \|^2 <\infty\}\quad  \mbox{ and } \quad D(C_s^\dagger) = \{u: \, \sum\limits_{k=0}^\infty \| C_s^\dagger \, \Pi_{k}\, u \|^2 <\infty\}.
\]
Note that  when $\sigma >1$ both these sets are at least as large as $ D(\hat{N}^{1/2})$. Note also that operators $ C_s  + C_s^\dagger$  and $ C_s  - C_s^\dagger$ are defined on $D(C_s) \cap D(C_s^\dagger)$.
The highlight of the next theorem is the fact that, indeed, operators $C_s^\dagger$ and $C_s$ are mutually adjoint, whenever $\sigma >1$. To carry out the proof we will temporarily denote the adjoint of operator $A$ by $A^*$.

\begin{corollary} \label{cor_C_closed}
  Whenever $\sigma >1$ the closed operators $C_s^\dagger$ and $C_s$ are mutually adjoint, i.e., $C_s = \left(C_s^\dagger \right)^*$ and $C_s^* = C_s^\dagger $.   In particular, the closure of operator $ C_s  + C_s^\dagger$ is self-adjoint and the closure of operator  $ C_s  - C_s^\dagger$ is anti self-adjoint.
\end{corollary}
\begin{proof}
  Assuming $\sigma >1$, (\ref{C_norms}) and (\ref{adjoint}) imply that for all $u \in D(C_s), v \in D(C_s^\dagger)$, we have
    \begin{equation}\label{adjoint_inD_of_N_sqrt}
     \langle \,  u \, | \, C_s^\dagger \, v\,\rangle = \sum_{k = 0}^{\infty}  \langle \,  \Pi_k\, u \, | \, C_s^\dagger \, v\,\rangle =  \sum_{k = 1}^{\infty}  \langle \,  C_s \Pi_k\, u \, | \, \Pi_{k-1} \, v\,\rangle =
      \langle \, C_s \, u \, |\, v\,\rangle
           \end{equation}
Thus, $\left(C_s^\dagger \right)^*$ is an extension of $C_s$, i.e., $C_s \subset \left(C_s^\dagger \right)^*$. It remains to show that  $ \left(C_s^\dagger \right)^* \subset C_s$. To this end, let $u \in D\left(\left(C_s^\dagger \right)^*\right)$. We will demonstrate that this implies $u \in D(C_s)$. First, by definition of $D\left(\left(C_s^\dagger \right)^*\right)$, the linear map
\begin{equation}\label{map}
v \mapsto \langle\, u \, | \, C_s^\dagger \, v\,\rangle
\end{equation}
is bounded on $D(C_s^\dagger )$. Thus, for any $v \in D(C_s^\dagger )$, we have
\[
 \langle\, u \, | \, C_s^\dagger \, v\,\rangle = \lim\limits_{n\rightarrow\infty}
 \langle\, u \, | \, C_s^\dagger \, \sum_{k=0}^{n} \Pi_k v\,\rangle
 = \lim\limits_{n\rightarrow\infty} \sum\limits_{k=0}^{n}
 \langle\, \Pi_{k+1}\, u \, | \, C_s^\dagger \, \Pi_k\, v\,\rangle
  = \lim\limits_{n\rightarrow\infty} \sum\limits_{k=0}^{n}
 \langle\, C_s \, \Pi_{k+1}\, u \, | \, \Pi_k\, v\,\rangle.
\]
  Since $D(C_s^\dagger )\subset \ell_2(\mathbb{N})$ is a dense subset and map (\ref{map}) is bounded, the limit on the right remains bounded for all $v$ in the entire Fock space $\ell_2(\mathbb{N})$. Therefore, the vector $w$ determined by $\Pi_k\, w = C_s \, \Pi_{k+1}\, u $ for all $k = 0, 1, \ldots$  is in the dual space of the Fock space, which is exactly itself.
This means that $u \in D(C_s)$, which is what we have set out to prove. Thus, $ \left(C_s^\dagger \right)^* = C_s$ and, since $C_s^\dagger$ is the maximal extension (of itself), also $ C_s^\dagger = \left(C_s^\dagger \right)^{**} = C_s^*$.

\end{proof}
\vspace{.2cm}

\noindent \emph{Remark.} The operators $C_s^\dagger, C_s$ are somewhat reminiscent to the standard operators of the ideal quantum gas, denoted $a^*(f)$, $a(f)$, see e.g. \cite{Merkli}. However, they are not the same. To see the difference, observe that even if $f(p) = p^{-s^*}$, we have
\[
C_s^\dagger \, |p^n\rangle = \sqrt{n+1} p^{-s^*}\, |p^{n+1} \rangle+ \sum_{q\neq p} q^{-s^*}\, |q p^n\rangle, \quad \mbox{ whereas }
\quad a^*(f) \, |p^n\rangle = \sqrt{n+1}\sum_q q^{-s^*}\, |qp^n\rangle.
\]
In a way, $C_s^\dagger$ is more discerning of the quality of states than the other operator. That is because, the boson sites are distinguishable unlike the gas particles.
\vspace{.2cm}

Next, we define a generalization of the displacement operator to the case of infinitely many boson sites; namely:
\begin{equation}\label{D_s}
  D_s = \exp \left( C_s^\dagger - C_s \right).
\end{equation}
It follows from Theorem \ref{theorem_C} pt. 4 that $D_s$ is a unitary operator.
We have the following fundamental
\begin{theorem}\label{theor_D} Let $D_s$  be as in (\ref{D_s}). Then
 \begin{equation}\label{D_s_on1}
  D_s \, |1\rangle = |s^*\rangle \quad \mbox{ for all } s \mbox{ with } \sigma = \Re s > 1.
\end{equation}
Note that an analogous result for $1/2< \Re s\leq 1$ is not claimed.
\end{theorem}

\begin{proof}

Note that
\[
\left[C_s^\dagger, -C_s  \right] =\left[  \sum_{p} p^{-s^*}\, \hat{a}_p^\dagger \, , \, -\sum_{q}q^{-s} \hat{a}_q \right] = - \sum_{p,q} p^{-s^*} q^{-s} [ \hat{a}_p^\dagger\, , \, \hat{a}_q] = \sum_{p} p^{-2\sigma} = P_1(2\sigma).
\]
Recall the Baker-Campbell-Hausdorff formula; namely, for two Hilbert-space operators $A, B$, one has
\[
e^Ae^B = e^{A+B + \frac{1}{2}[A,B]}.
\]
Setting $A = C_s^\dagger$ and $B = - C_s$, we obtain an alternative representation of $D_s$, namely:
\begin{equation}\label{D_s_prime}
  D_s = e^{-P_1(2\sigma)/2}\,\exp \left( C_s^\dagger \right) \,\exp\left( - C_s\right).
\end{equation}
Note that $(\hat{a}_p^\dagger)^n \, |1\rangle = \sqrt{n!} \, |p^n\rangle $, and so
\[
 \exp \left(  p^{-s^*} \, \hat{a}_p^\dagger \right)\, |1\rangle = \sum_{n=0}^\infty  \frac{p^{-ns^*}}{n!} \, (\hat{a}_p^\dagger)^n \, |1\rangle = \sum_{n=0}^\infty  \frac{p^{-ns^*}}{\sqrt{n!}} \, |p^n\rangle.
\]
Also,
\[
 \exp \left(  q^{-s^*} \, \hat{a}_q^\dagger \right)\, \exp \left(  p^{-s^*} \, \hat{a}_p^\dagger \right)\,\, |1\rangle =  \sum_{m,n}  \frac{\left(q^{m}p^{n}\right)^{-s^*}}{\sqrt{m!}\sqrt{n!}} \,\,\, |q^mp^n\rangle, \ldots \mbox{ etc. }
\]
It follows that
\[
\exp \left(C_s^\dagger \right) \,\, |1\rangle =\exp \left( \sum_{p} p^{-s^*}\hat{a}_p^\dagger \right) \,\, |1\rangle = \prod_p \exp( p^{-s^*}\hat{a}_p^\dagger ) \, |1\rangle = \sum_k \frac{k^{-s^*}}{x_k} \, |k\rangle.
\]
Since, $ \exp\left( - \sum_{p}p^{-s} \hat{a}_p \right)\, |1\rangle = |1\rangle$, (\ref{D_s_prime}) implies
\begin{equation}\label{Ds_of1}
  D_s \, |1\rangle = e^{-P_1(2\sigma)/2}\,\sum_k \frac{k^{-s^*}}{x_k} \, |k\rangle = |s^* \rangle.
\end{equation}
 This completes the proof.
\end{proof}
\vspace{.5cm}

\noindent
\emph{Remark.} Formula (\ref{D_s_on1}) is a generalization of a well known formula that holds for a regular coherent state. % In contrast to results obtained via the QFT approach, this result is rigorous.

\section{Hamiltonians}\label{section_Hamilt}

We consider several types of observables whose expectation can be evaluated in some explicit format whenever the system consisting of bosons placed on the infinite array of sites is in a NCS. This is a natural question from the physics standpoint. Indeed, NCS, being the eigenstates of the annihilation operators, are the natural candidates for the most robust states of a quantum system. Indeed, this argument is already classical with regards to regular CS,  \cite{Zurek}, and may be expected to hold in a more general context, such as the NCS. The discussion also highlights the special role of the Dirichlet ring.

\subsection{Expectation of the Bose-Hubbard Hamiltonian with all possible hopping terms}\label{subsection_number_op}

First observe that since $\hat{N}_p = \hat{a}_{p}^\dagger \hat{a}_{p}$, we have
\begin{equation}\label{s_particle_number_better}
  \langle s | \, \hat{N}\, | s\rangle = P_1(2\sigma).
\end{equation}
In particular, this shows that every coherent state caries a finite number of particles.
Also, since
\[
\hat{N}_p^2 = \hat{a}_{p}^\dagger \hat{a}_{p}\hat{a}_{p}^\dagger \hat{a}_{p}= \hat{a}_{p}^\dagger \hat{a}_{p}^\dagger \hat{a}_{p} \hat{a}_{p} + \hat{a}_{p}^\dagger \hat{a}_{p},
\]
we have
\begin{equation}\label{s_particle_number_sq}
  \langle s \,| \, \sum_{p}\hat{N}_p^2\, |\, s \rangle = P_1(2\sigma) + P_1(4\sigma).
\end{equation}
Recall the Bose-Hubbard Hamiltonian has the form
\begin{equation}\label{NT_BH}
  \mathcal{H}= \sum_{p}  \frac{U}{2}\, \hat{N}_{p}(\hat{N}_{p} - 1) - \mu\, \hat{N}_{p} - \tau\,\sum_{p,q} \, \hat{a}_{p}^\dagger \hat{a}_{q} + \hat{a}_{q}^\dagger \hat{a}_{p},
\end{equation}
It follows from a more general statement (\ref{Hamilt_n_N}), proven below, that this Hamiltonian commutes with the number operator.
One usually assumes that the summation in the hopping term is for nearest neighbor only. However, for our purposes it is easier to assume that hopping is equally probable for any pair $p,q$. In such a case,
\begin{equation}\label{s_hopping}
  \langle s | \,\,\sum_{p,q} \, \hat{a}_{p}^\dagger \hat{a}_{q} + \hat{a}_{q}^\dagger \hat{a}_{p}\,\, | s\rangle
  = 2\,|P_1(s)|^2,
\end{equation}
where one needs to assume $\Re s > 1$ to obtain convergence.
Collecting these findings, we obtain a formula for the expectation of the full Hamiltonian

\begin{equation}\label{s_BH}
  \langle s | \, \mathcal{H}\, | s\rangle = \frac{U}{2}\, P_1(4\sigma) -  \mu\, P_1(2\sigma) - 2\tau  \, |P_1(s)|^2, \quad \sigma > 1.
\end{equation}
 Thus,  the zeros of $P_1$ mark those NCS where the amplitude of disorder $\tau$ is not affecting the system's energy.  However, little is known about the number or position of the zeros of $P_1(s)$ in general, \cite{Froberg}. We know that when $\sigma$ is large enough, then $|2^{-s}| > |\sum_{p\geq 3} p^{-s}|$, so $P_1(s)$ cannot have any zeros on the right from a certain vertical cut-off line.

\subsection{Expectation formulas for a broader family of Hamiltonians}

A rich theory should be able to accommodate an ample selection of Hamiltonian operators. These are self-adjoint operators that commute with the number operator. We will demonstrate that the framework at hand possesses this feature. We start by calculating a few commutators. We will use the following well-known facts (easily proven by induction):
\[
[\, \hat{a}_p^n\, , \, \hat{a}_p^\dagger] = n \hat{a}_p^{n-1},\quad [\, \hat{a}_p \, , (\hat{a}_p^\dagger )^n\,] = n (\hat{a}_p^\dagger)^{n-1}\quad \mbox{ for all } n \geq 1.
\]
Utilizing these identities, and the definition (\ref{aks}), we obtain
\[
\begin{split}
[\, \hat{a}_k^\dagger\hat{a}_n\, , \, \hat{a}_p^\dagger\hat{a}_p\,]& =
 \hat{a}_k^\dagger\hat{a}_n\hat{a}_p^\dagger\hat{a}_p - \hat{a}_p^\dagger\hat{a}_p\hat{a}_k^\dagger\hat{a}_n \\ & \\
 & = \hat{a}_k^\dagger \left[ \, \hat{a}_p^\dagger\hat{a}_n + a_p(n) (\hat{a}_p)^{a_p(n) -1} \prod_{q\neq p} (\hat{a}_q)^{a_q(n)}  \,\right]\hat{a}_p
- \hat{a}_p^\dagger \left[ \, \hat{a}_k^\dagger\hat{a}_p + a_p(k) (\hat{a}_p^\dagger)^{a_p(n) -1} \prod_{q\neq p} (\hat{a}_q^\dagger)^{a_q(n)}  \,\right]\hat{a}_n \\ &\\
& = \hat{a}_{kp}^\dagger \hat{a}_{np} + a_p(n) \hat{a}_{k}^\dagger \hat{a}_{n}
- \hat{a}_{kp}^\dagger \hat{a}_{np} - a_p(k) \hat{a}_{k}^\dagger \hat{a}_{n} \\ \\
& = \left(a_p(n) - a_p(k) \right)\, \hat{a}_{k}^\dagger \hat{a}_{n} .
\end{split}
\]
Thus, for arbitrary $h_{n,k} \in \mathbb{C}$ we have
\begin{equation}\label{Hamilt_n_N}
  [\, h_{n,k}\, \hat{a}_n^\dagger\hat{a}_k  +
  h_{n,k}^*\, \hat{a}_k^\dagger\hat{a}_n\, , \, \hat{N}\, ]
  = h_{n,k}\, \left[\Omega(k) - \Omega(n) \right]\, \hat{a}_{n}^\dagger \hat{a}_{k} +
  h_{n,k}^*\, \left[\Omega(n) - \Omega(k) \right]\, \hat{a}_{k}^\dagger \hat{a}_{n}.
\end{equation}
It follows that hopping terms of the form $\hat{a}_k^\dagger\hat{a}_n + \hat{a}_n^\dagger\hat{a}_k$ commute with the number operator, provided $\Omega(n) = \Omega(k)$. Therefore, we have a broad family of Hamiltonians commuting with the number operator in the form:
\begin{equation}\label{Hamilt}
  \hat{H} = \sum_{n}\, \mbox{Poly}\left(: \, \hat{a}_n^\dagger\hat{a}_n\, :\right)\,\,  +  \sum_{n\neq k: \Omega(n)= \Omega(k)}  \left(\, h_{n,k}\, \hat{a}_n^\dagger\hat{a}_k  +
  h_{n,k}^*\, \hat{a}_k^\dagger\hat{a}_n\, \right),
\end{equation}
where $ \mbox{Poly}$ is any polynomial with real coefficients, and the powers of its argument $\hat{a}_n^\dagger\hat{a}_n$ are Wick ordered. All such operators $\mathcal{H}$ commute with the number operators. When the system is in an NCS, the expectation of such a Hamiltonian is:
\begin{equation}\label{s_Hamiltonian}
  \langle s | \, \hat{H}\, | s\rangle = \sum_{n} \mbox{Poly}\left(\, n^{-2\sigma}\,\right) \,\, +
  \sum_{n\neq k: \Omega(n)= \Omega(k)}  \left(\, h_{n,k}\,\,  n^{-s^*} k^{-s}  + c.c. \right) .
\end{equation}
\vspace{.5cm}

\noindent
\emph{Example.} Consider a generalization of the hopping term in (\ref{NT_BH}), which is given by
\begin{equation}\label{Hamilt_spec_zeta}
  \hat{H}_N =   \sum_{k,n: \Omega(k) = \Omega(n)\leq N } \, \hat{a}_n^\dagger\hat{a}_k,
\end{equation}
  It is seen by inspection that
\begin{equation}\label{s_Hamiltonian_spec_zeta}
  \langle s | \, \hat{H}_N\, | s\rangle =  1 + |P_1(s)|^2+ |P_2(s)|^2+ |P_3(s)|^2 +\ldots + |P_N(s)|^2 .
\end{equation}
%We do not know if the expression above will give a finite value for any $s: \Re s > 1$ when $N \rightarrow \infty$.

\subsection{The Dirichlet ring represented via the annihilation operators}\label{subsection_Dirichlet_ring}

We outline a method which leads to additional Hamiltonians that can be analyzed using the NCS. This construction will come to full fruition in Section \ref{section_FT}.
Recall that the algebra of formal Dirichlet series consists of all series of the form
$
\tilde{f}(s) = \sum_k f(k) k^{-s} .
$
where $f: \mathbb{N}\rightarrow \mathbb{C}$ is arbitrary. As is well known,
$
 \tilde{f}(s)  \tilde{g}(s) = \tilde{h}(s),
$
where $h = f\star g$ (the Dirichlet convolution of $f$ and $g$), is defined via
$
h (k) = \sum_{d|k} f(d) g(k/d).
$
With this understood, we introduce operators of the form:
\begin{equation}\label{Df_def}
 \hat{F}_s =f(1)+ f(2)\, 2^{-s}\,\hat{a}_2 + f(3)\, 3^{-s} \, \hat{a}_3 + f(4)\, 4^{-s}\, \hat{a}_4 + \ldots .
\end{equation}
Note that operator $C_s$ defined in (\ref{def_Cs}) is of this type with $f$ being the indicator function of the primes. %Also,  $c\hat{D}_f = \hat{D}_{cf}$ for any constant $c$.
Identity (\ref{akan}) implies that
\begin{equation}\label{Zfprod}
  \hat{F}_s \hat{G}_s = \hat{G}_s \hat{F}_s = \hat{H}_s \quad \mbox{ where } h = f \star g.
\end{equation}
Thus, construction (\ref{Df_def}) gives yet another representation of the Dirichlet ring.
Identity (\ref{s_is_eigen4n}) readily implies
\begin{equation}\label{D_f_on_s}
\hat{F}_{s'}\, |s\rangle = \tilde{f}(s+s')\, |s\rangle .
\end{equation}
A more general identity is derived below, see (\ref{Df_on_Us}). We defer a discussion how operators of type $F_s$ may be utilized in construction of Hamiltonians, with calculable expectation when the system is in NCS, to Section \ref{section_FT}.

\section{The special unitary operator and the Fourier duality for the multiplicative group of positive rationals}\label{section_FT}

In this section, we highlight the role of the multiplicative group of positive rationals, $\mathbb{Q}_+$. The general theory of abstract harmonic analysis has textbook expositions, e.g., \cite{Rudin}, \cite{Folland}. In its $\mathbb{Q}_+$ -specialized form, it has been applied in number theory, see \cite{Elliott}. However, to our knowledge, it has not been applied in quantum theory until the authors introduced it in \cite{Sowa_Fransson} within the context of boson theory. Specifically, it has been shown that the transform provides a bridge between the number-theoretic framework, on one hand, and the Kastrup model of the harmonic oscillator in $H_2(\mathbb{T}^1)$, see \cite{Kastrup}, on the other. Here, we demonstrate how the transform enables the construction of a far-reaching generalization of the NCS introduced in Section \ref{section_construction}.

\subsection{The generalized Fourier transform}

The unitary map introduced in Subsection \ref{seubsection_Uf} is related to the Fourier transform on the group of positive rationals $\mathbb{Q}_+$. Recall that the Pontryagin dual group is an infinite torus
\[
\mathbb{T}^\omega = \hat{\mathbb{Q}}_+ = \prod_{p \in \mathcal{P}} U(1).
\]
Note that the dimensions of the torus (represented by copies of $U(1)$) are indexed by primes. Also, it is natural to equip $\mathbb{T}^\omega$ with the global coordinates $\vec{\mu} = (\mu_2, \mu_3, \mu_5, \ldots )$.

Recall that $\mathbb{Q}_+$ is equipped with discrete topology and discrete measure.  On the other hand, $\mathbb{T}^\omega$ is equipped with the product topology, making it into a compact space. It is also endowed with the bi-invariant probabilistic measure $d\vec{\mu} = d\mu_2\, d\mu_3 \, d\mu_5 \ldots$.  In addition $d\vec{\mu}$ is a Borel measure, which satisfies
\[
d\vec{\mu} \left( (\alpha_2, \beta_2]  \times  \ldots \times (\alpha_p, \beta_p] \times (0,1]  \times (0,1] \times \ldots  \right) =
|\beta_2 - \alpha_2| \ldots |\beta_p - \alpha_p|.
\]
The set of point measures $\delta_w, w\in \mathbb{Q}_+$ furnishes an distinct orthonormal basis in $\ell_2(\mathbb{Q}_+)$. On the other hand $L_2(\mathbb{T}^\omega, d\vec{\mu})$ has a distinct orthonormal basis consisting of functions $\exp 2\pi i \vec{w}\cdot \vec{\mu}$, where  $\vec{w} = (a_2(w), a_3(w), a_5(w), \ldots)$, and $a_p(w)$ are determined by the prime decomposition
\[
w = \prod_{p\in \mathcal{P}} p^{a_p}, \, a_p \in\mathbb{Z}.
 \]
The corresponding Fourier transform exchanges these distinct basis functions, namely
\begin{equation}\label{FT_simple}
\mbox{FT}: \delta_w \leftrightarrow e^{2\pi i \vec{w}\cdot \vec{\mu}}.
\end{equation}
In other words, for a function $f = \sum_{w\in \mathbb{Q}_+} z_w \delta_w $, one has its transform in the form
\[
\hat{f}(\vec{\mu}) = \sum_{w\in \mathbb{Q}_+} z_w e^{2\pi i \vec{w}\cdot \vec{\mu}}.
\]
A special role is played by the subspace $\ell_2(\mathbb{N})\subset \ell_2(\mathbb{Q}_+) $. Let $P_+: \ell_2(\mathbb{Q}_+) \rightarrow \ell_2(\mathbb{N})$ be the orthogonal projection. The Fourier duality endows the complementary subspace $H_2(\mathbb{T}^\omega, d\vec{\mu})$. It consists of functions whose only nonzero coefficients are those whose index is a natural number, i.e., the coefficients with fractional indices vanish. The corresponding orthogonal projection will be denoted by the same symbol.
The following diagram captures the essential features of the two mirroring structures:
\begin{equation}\label{diagram}
\begin{array}{lccc}
P_+:& L_2(\mathbb{T}^\omega, d\vec{\mu}) & \rightarrow & H_2(\mathbb{T}^\omega, d\vec{\mu})\\
& \\
 & \big\updownarrow\mbox{FT} & & \big\updownarrow\mbox{FT} \\
  & \\
 P_+:& \ell_2(\mathbb{Q}_+) & \rightarrow & \ell_2(\mathbb{N})
\end{array}
\end{equation}
When the Fourier transform is restricted to $ \ell_2(\mathbb{N})$ it assumes the form
\begin{equation}\label{FT_simple2}
\mbox{FT}: \delta_n = | n\rangle \leftrightarrow e^{2\pi i \vec{n}\cdot \vec{\mu}}.
\end{equation}

\subsection{NCS in the Fourier-dual picture}\label{seubsection_NCS_F_Dual}

Applying the transform (\ref{FT_simple2}) to the NCS (\ref{the_s_all}), we obtain its Fourier-dual representation, which we will also denote $|s\rangle$, namely:
\begin{equation}\label{the_s_F_dual}
   |s\rangle = e^{-P_1(2\sigma)/2}\, \sum_k \frac{k^{-s}}{x_k}\, e^{2\pi i \vec{k}\cdot \vec{\mu}}.
\end{equation}
The following factorization is immediate:
\begin{equation}\label{the_s_factorized}
  |s\rangle = e^{-P_1(2\sigma)/2}\, \prod_{p}\,\sum_{n=0}^{\infty} \frac{p^{-n s}}{\sqrt{n !}}\, e^{2\pi i n \mu_p}.
\end{equation}
This is a reformulation of (\ref{inf_tens_pr}) that replaces the infinite tensor product with the ordinary infinite product.

Furthermore,   in the Fourier-dual representation we have
\begin{equation}\label{N_p_F_dual}
  \hat{N}_p = \frac{1}{2\pi i} \frac{\partial}{\partial \mu_p}.
\end{equation}
It follows that
\begin{equation}\label{N_p_comm_exp}
  [\,  \hat{N}_p\, , \, e^{2\pi i \mu_p} \,] = e^{2\pi i \mu_p},
\end{equation}
where $e^{2\pi i \vec{k}\cdot \vec{\mu}}$ is the multiplication operator. Equivalently, the above may be stated in the form
\begin{equation}\label{N_p_N_p_plus}
e^{2\pi i \mu_p} \, (\hat{N}_p + 1)   \, e^{-2\pi i \mu_p}  = \hat{N}_p.
\end{equation}
The following identities have been demonstrated in \cite{Sowa_Fransson} (and can easily be verified directly):
\begin{equation}\label{a_p_F_dual}
  \hat{a}_p = (\hat{N}_p+1)^{-1/2}\, e^{-2\pi i \mu_p} \, \hat{N}_p, \quad\quad
   \hat{a}_p^\dagger = \hat{N}_p\, e^{2\pi i \mu_p} \, (\hat{N}_p+1)^{-1/2}.
\end{equation}
These relations are a form of the Holstein-Primakoff transform, \cite{Holstein-Primakoff}. Note that these relations, as well as the factorization (\ref{the_s_factorized}), facilitate calculations involving the action of these common operators on $|s\rangle$.

\subsection{A special unitary operator and generalized NCS}\label{seubsection_Uf}

The NCS representations given by (\ref{the_s_all}) and (\ref{the_s_F_dual}) may be welded together to yield generalized nonlocal coherent states. To this end, we introduce the special unitary map $U_{\vec{\mu}}: \ell_2(\mathbb{N}) \rightarrow \ell_2(\mathbb{N})$, given by
\begin{equation}\label{def_U_mu}
  U_{\vec{\mu}} = \exp \left(2\pi i \sum_{p\in \mathcal{P}} \mu_p \hat{N}_p\right).
\end{equation}
Since $\hat{N}_p \, |k\rangle = a_p(k)\, |k\rangle$, we have $  U_{\vec{\mu}}\, |k\rangle = e^{2\pi i \vec{k} \cdot\vec{\mu} } \, |k\rangle$.
%\[
% \hat{N}_p |s\rangle = e^{-P_1(2\sigma)/2} \sum \frac{a_p(k)\, k^{-s}}{x_k}\, |k\rangle .
%\]
Thus,
\begin{equation}\label{Uf_basics}
 U_{\vec{\mu}} \, |s\rangle = e^{-P_1(2\sigma)/2}  \sum \frac{k^{-s}}{x_k}e^{2\pi i \vec{k} \cdot\vec{\mu} } \, |k\rangle.
\end{equation}
This defines the first generalization of the NCS, namely, $|s, \vec{\mu}\rangle := U_{\vec{\mu}} \, |s\rangle$. Note that by (\ref{an_onk}):
\begin{equation}\label{calc_n_mu}
\hat{a}_n \,  \sum \frac{k^{-s}}{x_k}e^{2\pi i \vec{k} \cdot\vec{\mu} } \, |k\rangle
  =  \sum \frac{k^{-s}}{x_k}e^{2\pi i \vec{k} \cdot\vec{\mu} } \, \frac{x_k}{x_{k/n}} \, |\, k/n\rangle
  =  n^{-s} e^{2\pi i \vec{n} \cdot \vec{\mu}} \,\, \sum \left(\frac{k}{n}\right)^{-s}\frac{1}{x_{k/n}} e^{2\pi i (\vec{k}-\vec{n})\cdot\vec{\mu}} \,\, |\, k/n\rangle .
\end{equation}
Therefore,
\begin{equation}\label{an_on_Uf_s}
 \hat{a}_n \,    |s, \vec{\mu}\rangle :=    n^{-s} e^{2\pi i  \vec{n}\cdot \vec{\mu}} \,\,   |s, \vec{\mu}\rangle
\end{equation}
In particular,
\begin{equation}\label{more_eigvals}
   \langle s, \vec{\mu}\, | \, \, \sum_{n}\, \mbox{Poly}\left(: \, \hat{a}_n^\dagger\hat{a}_n\, :\right)\, \,  |s, \vec{\mu}\,\rangle = \sum_{n} \mbox{Poly}\left(\, n^{-2\sigma}\,\right).
\end{equation}
Also, for operators $  \hat{F}_s$ as in (\ref{Df_def}) we have:
\begin{equation}\label{Df_on_Us}
  \hat{F}_{s'}\,  |s, \vec{\mu}\,\rangle =\left(\sum f(n) \,  n^{-(s+s')}\, e^{2\pi i \vec{n}\cdot \vec{\mu}} \right) \,  |s, \vec{\mu}\,\rangle .
\end{equation}
This identity reduces to (\ref{D_f_on_s}) when $\vec{\mu} =0$, i.e., $U_{\vec{\mu}} = I$. Thus,
\begin{equation}\label{Dirichlet_interaction}
  \langle s, \vec{\mu}\, | \, \hat{F}_{s'}^\dagger\,\hat{F}_{s'}\,  |s, \vec{\mu}\,\rangle = \left|\sum f(n) \,  n^{-(s+s')}\, e^{2\pi i \vec{n}\cdot \vec{\mu}}  \right|^2.
\end{equation}
Formulas (\ref{more_eigvals}), (\ref{Dirichlet_interaction}) provide explicit expectation for systems in the NCS with a Hamiltonian of the form
\begin{equation}\label{_HU}
 \mathcal{H}_{f,s',\vec{\mu}, \,\mbox{\small Poly}} =  \,\left( \sum_{n}\, \mbox{Poly}\left(: \, \hat{a}_n^\dagger\hat{a}_n\, :\right) +  \hat{F}_{s'}^\dagger\,\hat{F}_{s'}\, \right)
\end{equation}
when the system is in the state $|s, \vec{\mu}\,\rangle$.
Of course, formal properties of such Hamiltonians, including the self-adjointness, will depend on the choice of $f$. This can be nontrivial, as we have seen in the case of operator $C_s$ (where the corresponding $f$ is the indicator function of the primes). Nevertheless there are a plethora of additional examples where it is not too difficult to establish these properties, e.g. whenever $f$ is finitely supported.

\subsection{The special unitary operator and the Fourier duality}

We point out that the unitary operators $U_{\vec{\mu}}$ is closely related to the Fourier duality. To see this, we need to involve a linear functional $\ell_1 : \ell_2(\mathbb{N})\rightarrow \mathbb{C}$, defined by
\[
\ell_1 \left( \, \sum_{n\in\mathbb{N}} c_n \delta_n \,\right) = \sum_{n\in\mathbb{N}} c_n.
\]
Clearly $\ell_1$ is not a bounded functional. Note that in $\ell_2(\mathbb{N})$ we have two ways of referring to the distinct basis, i.e., $\delta_k = |k\rangle, k \in \mathbb{N}$. It follows from the first basic identity (\ref{Uf_basics}) that
\[
  \ell_1 \left(U_{\vec{\mu}} \, \, \sum_{n\in\mathbb{N}} c_n \delta_n \right)  = \sum_{n\in\mathbb{N}} c_n \exp(2\pi i \vec{n}\cdot \vec{\mu}).
\]
 In other words,
 \begin{equation}\label{FT_via_Umu}
 \ell_1 \circ  U_{\vec{\mu}}\, ( f ) = \hat{f}(\vec{\mu}) \quad \mbox{ for all } \quad   f\in \ell_2(\mathbb{N}).
 \end{equation}
In this way the Fourier transform on $\mathbb{Q}_+$ arises from the underlying quantum theoretic framework.

\section{The second generalization of NCS and the resolution of identity} \label{section_res_I}

A closer examination reveals that neither $|s\,\rangle$ nor $|s, \vec{\mu}\,\rangle$ admit a resolution of the identity formula that would be analogous to the well known formula for the regular coherent states. However, such a formula can be obtained if one generalizes the NCS further. To this end,  $s$ is now fixed but arbitrary with the requirement  $\sigma = \Re s > 1/2$. Furthermore,  we introduce variable $\vec{r} = (r_2, r_3, r_5, \ldots) \in \mathbb{R}_+^\omega $, so that each $r_p\in (0,\infty)$ with the additional requirement
\begin{equation}\label{y_ps}
P_1(2\sigma, |\vec{r}|^2) = \sum_{p} p^{-2\sigma}r_p^2 <\infty.
\end{equation}
%The fact that $\mathbb{R}_+$ may be interpreted as a multiplicative group will not be exploited here.
We endow $\mathbb{R}_+^\omega$ with the product topology. Also, let $ d\chi_p = 2 \exp(-p^{-2\sigma}r_p^2) p^{-2\sigma} r_p \, dr_p$, so that
\begin{equation}\label{d_chi_p}
  \int_{0}^{\infty}\, d\chi_p = 1.
\end{equation}
Analogously to the case of the measure $d\vec{\mu}$ on $\mathbb{T}^\omega $ featured in Section \ref{section_FT}, we introduce a probabilistic Borel measure $d \vec{\chi}$ on $\mathbb{R}_+^\omega$, determined by the requirement
\[
d \vec{\chi} \left( (a_2, b_2]  \times  \ldots \times (a_p, b_p] \times (0, \infty)  \times (0,\infty) \times \ldots  \right) =
\int_{a_2}^{b_2} d \chi_2 \ldots \int_{a_p}^{b_p} d\chi_p,
\]
for arbitrary $0\leq a_2< b_2\leq \infty,  \,\ldots \, 0\leq a_p< b_p\leq \infty$.

We now define a second generalization of the NCS via
\begin{equation}\label{sxmu}
 |\, s, \vec{r}, \vec{\mu} \, \rangle = e^{-\frac{1}{2}  P_1(2\sigma, |\vec{r}|^2)} \sum_{k} \frac{k^{-s}}{x_k}\, \prod r_p^{a_p(k)}\,e^{2\pi i \vec{k}\cdot \vec{\mu}}\, |\,k\,\rangle .
\end{equation}
Note that
\[
\sum_{k} \frac{k^{-s}}{x_k}\, \prod r_p^{a_p(k)} \,e^{2\pi i \vec{k}\cdot \vec{\mu}}\, |\,k\,\rangle = \bigodot_{p} \sum_{n=0}^{\infty} \frac{\left(p^{-s}r_p e^{2\pi i \mu_p}\right)^n}{\sqrt{n!}} \, |\, p^n\rangle.
\]
Also,
\[
\sum_{k} \left|\frac{k^{-s}}{x_k}\, \prod r_p^{a_p(k)} \,e^{2\pi i \vec{k}\cdot \vec{\mu}}\,\right|^2 = \prod_{p} \sum_{n=0}^{\infty} \frac{\left(p^{-2\sigma}r_p^2\right)^n}{n!}  = \prod_{p} \exp{p^{-2\sigma}r_p^2} = \exp P_1(2\sigma, |\vec{r}|^2),
\]
which demonstrates that $ |\, s, \vec{r}, \vec{\mu} \, \rangle$ have norm $1$.

Next, we examine
\begin{equation}\label{res1}
   |\, s, \vec{r}, \vec{\mu} \, \rangle \langle \, s,\vec{r}, \vec{\mu}\, | =
   e^{- P_1(2\sigma, |\vec{r}|^2)} \sum_{k,l} \frac{k^{-s}l^{-s^*}}{x_kx_l}\, \prod r_p^{a_p(k)+a_p(l)}\,e^{2\pi i (\vec{k}-\vec{l})\cdot \vec{\mu}}\, \,|\,k\,\rangle\langle\, l\, |.
\end{equation}
Applying integration with respect to $d\vec{\mu}$ eliminates the off-diagonal terms, i.e.,
 \[
 \int_{\mathbb{T}^\omega }\, d\vec{\mu}\,\, |\, s, \vec{r}, \vec{\mu} \, \rangle \langle \, s, \vec{r}, \vec{\mu}\, | =
  e^{- P_1(2\sigma, |\vec{r}|^2)} \sum_{k} \frac{k^{-2\sigma}}{x_k^2}\, \prod r_p^{2a_p(k)}\, |\,k\,\rangle\langle\, k\, |.
 \]
 For every $k$, the term at $|\,k\,\rangle\langle\, k\, |$ is positive, and has the form of a product of finitely many factors of the form
 \[
 \exp(-p^{-2\sigma}r_p^2) \,\, (p^{a_p(k)})^{-2\sigma} r_p^{2a_p(k)}\, \frac{1}{a_p(k) !},
 \]
 and, in addition, infinitely many factors $ \exp(-q^{-2\sigma}r_q^2)$, where $q$ runs over all those primes that are not divisors of $k$.
Note that
\[
 \int_{0}^{\infty}  \exp(-p^{-2\sigma}r_p^2) \,\, (p^{-\sigma} r_p)^{2a_p(k)}\, \frac{1}{a_p(k) !} \,
 2 p^{-2\sigma} r_p \, dr_p = 1.
\]
This together with (\ref{d_chi_p}) implies that
\begin{equation}\label{res_identity}
  \int_{\mathbb{R}_+^\omega}\, d\vec{\chi} \, e^{P_1(2\sigma, |\vec{r}|^2)} \,\int_{\mathbb{T}^\omega}\,d\vec{\mu} \,\, |\, s, r, \vec{\mu} \, \rangle \langle \, s, r, \vec{\mu}\, | = \sum_{k} \, |\,k\,\rangle\langle\, k\, | = I.
\end{equation}
This is the coveted resolution of the identity formula.  It implies in particular that
\[
\{ |\,s, \vec{r}, \vec{\mu}\,\rangle: \vec{\mu}\in\mathbb{T}^\omega \mbox{ and } \vec{r} \in \mathbb{R}_+^\omega, \mbox{ such that } (\ref{y_ps}) \mbox{ holds} \}
\]
is an over-complete system in $\ell_2(\mathbb{N})$ for any $s$ with $\sigma >1/2$. Note that in (\ref{res_identity}) the parameter $s$ is arbitrary but fixed.
\vspace{.5cm}

\noindent
\emph{Remark 1.} Note that in a way we utilize the term  $\exp( -P_1(2\sigma, |\vec{r}|^2))$ in the definition of $d\vec{\chi}$, only to remove it later on via the factor $\exp(P_1(2\sigma, |\vec{r}|^2)$ present in the formula (\ref{res_identity}). The initial insertion is necessary to have a rigorous definition of $d\vec{\chi}$. The subsequent removal is necessary because the term $\exp( -P_1(2\sigma, |\vec{r}|^2))$ is already present as the normalizing factor of $|\, s, \vec{r}, \vec{\mu} \, \rangle \langle \, s, \vec{r}, \vec{\mu}\, |$.
\vspace{.5cm}

\noindent
\emph{Remark 2.}
Recall that the well-known resolution of the identity formula for the regular coherent states has the form $\pi^{-1} \int_{\mathbb{C}} |z\rangle\langle z| \, d^2 z = I$. This might raise hopes that one could find an analogous formula for the NCS based on integration of $|s\rangle\langle s|$ over the half plane $\{s: \Re s> 1/2\} $ with a suitable measure. However, such hope is futile, given that the integrands are of the form $k^{-2\sigma}/ x_k^2$ and the constants $x_k$ depend on the morphology of $k$ and not on its magnitude, e.g. the integral of $p^{-2\sigma}$ with respect to some measure on the $\sigma$ axis would have to be the same for all primes $p$, which is not possible. This difficulty is circumvented in  (\ref{res_identity}) by introducing the indispensable variables $r_p, \mu_p$.

\subsection{Generalized NCS in the complex notation}

A substitution $z_p = r_p \exp{(2\pi i \mu_p)}$ allows a more compact notation and facilitates calculations. In new notation (\ref{sxmu}) assumes the form
\begin{equation}\label{sz}
  |\, s, \vec{z}\, \rangle   = e^{-\frac{1}{2}  P_1(2\sigma, |\vec{z}|^2)} \sum_{k} \frac{k^{-s}}{x_k}\, \prod z_p^{a_p(k)}\, |\,k\,\rangle,
  \quad \mbox{ where }\quad
P_1(2\sigma, |\vec{z}|^2) = \sum_{p} p^{-2\sigma}|z_p|^2.
\end{equation}
Note that $z_p$ are nonzero but otherwise arbitrary complex numbers, albeit the sequence $z_2, z_3,z_5,\ldots$ is required to satisfy the constraint (\ref{y_ps}) or, equivalently:
\begin{equation}
    P_1(2\sigma, |\vec{z}|^2) = \sum_{p} p^{-2\sigma}|z_p|^2 =: \| \vec{z} \|_\sigma^2.
\end{equation}
Interestingly, the norms $\| \vec{z} \|_\sigma$ are analogous to the Sobolev norms, except the weights are introduced via the primes. 

In this notation, the resolution of the identity assumes the form
\begin{equation}\label{res_identity_z}
  \int_{\mathbb{R}_+^\omega}\, d\vec{\chi} \, e^{P_1(2\sigma, |\vec{z}|^2)} \,\int_{\mathbb{T}^\omega}\,d\vec{\mu} \,\, |\, s, \vec{z} \, \rangle \langle \, s, \vec{z}\, | = I .
\end{equation}

A straightforward calculation gives an inner product formula for the generalized NCS analogous to (\ref{prod}) in the form:
  \begin{equation}\label{prod_z}
     \langle s, \vec{z}\, |\, s', \vec{z'} \rangle = \exp \left(-P(2\sigma, |\vec{z}|^2)/2 -P(2\sigma', |\vec{z'}|^2)/2 + P(s^* + s', \vec{z^*}\cdot\vec{z'})\right) ,
  \end{equation}
  where $P(s^* + s', \vec{z^*}\cdot\vec{z'}) = \sum_{p} p^{-(s^*+s')} z_p^*z_p'$.
A direct calculation analogous to (\ref{calc_n_mu}) yields
  \begin{equation}\label{a_p_s_z}
    \hat{a}_p\, |\, s, \vec{z} \rangle = p^{-s} z_p \, |\, s, \vec{z} \rangle\quad \mbox{and, more generally, } \quad \hat{a}_n\, |\, s, \vec{z} \rangle = n^{-s} \prod_p  z_p^{a_p(n)} \, |\, s, \vec{z} \rangle.
  \end{equation}
    Since $\langle\, s, \vec{z}\, |\,\hat{a}_p^\dagger\hat{a}_p\, |\, s, \vec{z} \, \rangle = p^{-2\sigma} |z_p|^2$, we obtain a generalization of (\ref{s_particle_number_better}), namely:
\begin{equation}\label{Nsz}
\langle\, s, \vec{z}\, |\,\hat{N}\, |\, s, \vec{z} \, \rangle  =P_1(2\sigma, |\vec{z}|^2),
\end{equation}
i.e., the expected number of particles of the state $|\, s, \vec{z} \, \rangle$ depends on $\sigma$ as well as $\vec{z}$.
We also have the following
  \begin{corollary} \label{cor_spectrum_a_n}
  For all $n$ the spectrum of $\hat{a}_n$  is the entire complex plane $\mathbb{C}$.
  \end{corollary}
  \begin{proof}
    Fix $s$ with $\sigma > 1/2$. Note that only a finite collection of $a_p(n)$ are nonzero, and the corresponding $z_p$ in (\ref{a_p_s_z}) are arbitrary complex numbers other than $0$. Selecting $z_p$ one can set the eigenvalue $n^{-s} \prod_p  z_p^{a_p(n)}$ to any value other than $0$. The spectrum of $\hat{a}_n$ is the closure of this set, i.e., the entire complex plane.   \end{proof}

\subsection{The displacement operator for the generalized NCS}

We observe that all the statements of Section \ref{section_Cs} easily generalize, leading to the displacement operator that generates the generalized NCS. To make this evident, we introduce  functions analogous to (\ref{n_prime_zeta}), namely:
\begin{equation}\label{n_prime_zeta_gen}
  P_n(s, \vec{z}) = \sum_{k: \Omega(k) = n}  k^{-s}\prod_p z_p^{a_p(k)}, \quad  \sigma > 1.
\end{equation}
when $s=2\sigma$ and $z_p$ are replaced with $|z_p|^2$, we will use notation  $P_n(2\sigma, |\vec{z}|^2)$ consistently with $P_1(2\sigma, |\vec{z}|^2)$ introduced above. We easily obtain an analog of (\ref{P_2_to_P_1}), namely
\begin{equation}\label{P_2_to_P_1gen}
  P_2(s, \vec{z}) = \frac{1}{2}\left( P_1(s, \vec{z})^2 + P_1(2s, \vec{z}) \right).
\end{equation}
With this understood Theorem \ref{theorem_C} holds  when the operators $C_s^\dagger, C_s$ are replaced with their analogues:
  \begin{equation}\label{def_Csgen}
      C_{s, \vec{z}} = \sum_{p}p^{-s}z_p\, \hat{a}_p,\quad  
 C_{s, \vec{z}}^\dagger = \sum_{p}p^{-s^*}z_p^* \, \hat{a}_p^\dagger,
\end{equation}
and all $P_n(2\sigma)$ are simultaneously replaced by $P_n(2\sigma, |\vec{z}|^2)$. The proof carries over verbatim by replacing every $p^{-s}$ with $p^{-s}z_p$.  

Furthermore, Theorem \ref{theor_D} also generalizes directly, assuming the form
\begin{equation}
    \exp \left( C_{s, \vec{z}} ^\dagger - C_{s, \vec{z}} \right)\, |1\rangle = |s^*, \vec{z^*}\rangle.
\end{equation}
The generalized displacement operator $$D_{s, \vec{z}} =  \exp \left( C_{s, \vec{z}} ^\dagger - C_{s, \vec{z}} \right)$$ is unitary. 

\subsection{Differential representation of the creation and annihilation operators.}
  As a consequence of (\ref{a_p_s_z}) we obtain
  \begin{equation}\label{a_p_star_psi}
    \langle s, \vec{z}\, | \,\hat{a}_p^\dagger\, |\, \psi \,\rangle = p^{-s^*}z_p^*  \langle s, \vec{z}\,\, |\,\psi\, \rangle, \quad
    \mbox{ for any } |\, \psi\,\rangle \in \ell_2(\mathbb{N}).
  \end{equation}
  On the other hand, by (\ref{prod_z}), we have
  \[
  \langle s, \vec{z}\, | \,\hat{a}_p\, |\, s', \vec{z'} \rangle = p^{-s'}z_p'  \langle s, \vec{z}\, |\, s', \vec{z'} \rangle  =
  \left( p^{s^*} \partial_{z_p^*} + p^{-s} z_p /2 \right)  \langle s, \vec{z}\, |\, s', \vec{z'} \rangle.
  \]
  Together with (\ref{res_identity_z}) this implies
   \begin{equation}\label{a_p_psi}
    \langle s, \vec{z}\, | \,\hat{a}_p\, |\, \psi \,\rangle =   \left( p^{s^*} \partial_{z_p^*} + p^{-s} z_p /2 \right)  \langle s, \vec{z}\,\, |\,\psi\, \rangle, \quad
    \mbox{ for any } |\, \psi\,\rangle \in \ell_2(\mathbb{N}).
  \end{equation}
  The above formula can be simplified via judicious scaling of the product $\langle s, \vec{z}\,\, |\,\psi\, \rangle$. Namely, set
  \begin{equation}\label{complex_rep}
     f_{\psi}( s^*, \vec{z^*}) = e^{P(2\sigma, |\vec{z}|^2)/2} \, \langle s, \vec{z}\,\, |\,\psi\, \rangle.
  \end{equation}
  Then,
  \begin{equation}\label{ap-del}
f_{\hat{a}_p\psi}( s^*, \vec{z^*})  =  p^{s^*} \partial_{z_p^*} \, f_{\psi}( s^*, \vec{z^*}), \quad
    f_{\hat{a}_p^\dagger \psi}( s^*, \vec{z^*})  =  p^{-s^*} z_p^* \, f_{\psi}( s^*, \vec{z^*}).
  \end{equation}
  This yields an alternative representation of the creation and annihilation operators via operations on the variables $z_p$, namely:
  \begin{equation}\label{z_delz}
    \hat{a}_p \leftrightarrow  p^{s^*} \partial_{z_p^*}\quad \mbox{ and }\quad
  \hat{a}_p^\dagger\leftrightarrow  p^{-s^*} z_p^*.
  \end{equation}
 In this way, the resolution of the identity formula facilitates calculations.  Formulas (\ref{a_p_star_psi}) and (\ref{a_p_psi}), or the equivalent formula (\ref{ap-del}), are analogues of well known formulas for the regular coherent states.
\vspace{.5cm}

\noindent
\emph{Example.} The complex variable representation of states gives strong advantage on some types of calculations, even those including a finite system. Here, we will discuss the simplest case of a dimer, i.e. a state vector representing two particles placed on sites $2$ and $3$. In accordance with (\ref{z_delz}), $\hat{N}_{p} = z_p^* \partial_{z_p^*}$, while the hopping part of (\ref{NT_BH}) assumes the form
\[
  \mathcal{H}_i= - \tau\, \left[\, (2/3)^{-s^*} z_2^* \partial_{z_3^*} + (3/2)^{-s^*} z_3^* \partial_{z_2^*}\,\right].
\]
One can search for solutions in the form of linear combinations of terms $z_2, z_3$ (one-particle states), or  $(z_2^*)^2, \, z_2^*z_3^*, \, (z_3^*)^2$ (two-particle states), etc.

\subsection{Bose-Hubbard Hamiltonian}

The NCS allow us to gain insight into the solutions of Bose-Hubbard Hamiltonian with a finite number, say, $N$, of sites. Let the sites by labeled $p_1 = 2, p_2 = 3, p_3 = 5, \ldots , p_{N}$, so that the hopping part assumes the form $\mathcal{H}_i = -\tau D$, where
\begin{equation}\label{hopper_n}
  D = \sum\limits_{j =1}^{N-1} \, a_j \, z_{p_j}^* \partial_{z_{p_{j+1}}^*} + \, a_j^{-1} z_{p_{j+1}}^* \partial_{z_{p_j}^*}, \quad \mbox{ where } a_j = \left(\frac{p_j}{p_{j+1}}\right)^{-s^*}.
\end{equation}
Crucially, $D$ satisfies the Leibnitz rule. Here we are only interested in this property in the space of polynomials in $z_j$. Thus, for any two polynomial $f, g$, we have
\begin{equation}\label{Leibnitz}
  D \, [fg] = D[f]\, g + f\, D[g].
\end{equation}
Note that the $n$-particle space has a basis that consists of monomials $z_2^{\alpha_2} z_3^{\alpha_3} \ldots z_{p_N}^{\alpha_{p_N}}$, where $\alpha_p$ are nonnegative integers, and their sum is $n$. Observe that each multi-index $\vec{\alpha} = (\alpha_ 2, \,  \alpha_3, \, \ldots ,\, \alpha_{p_N})$, where $|\vec{\alpha}| := \alpha_ 2 + \alpha_3 + \ldots + \alpha_{p_N} = n$, is also a partition of the integer $n$ into $N$ parts. In fact, we will demonstrate that the eigenvalues of $D$ are determined by the partitions. To this end, let us observe that the matrix od $D$ in the $1$-particle space assumes the form
\begin{equation}\label{D_matrix}
 \left(
    \begin{array}{ccccc}
       &  a_1 &  &  &  \\
      a_1^{-1} &  & a_2 &  &  \\
       & a_2^{-1} &  & \ddots &  \\
       &  & \ddots &  & a_{N-1} \\
       &  &  & a_{N-1}^{-1} &  \\
    \end{array}
  \right)
\end{equation}
It is easily seen that the eigenvalues do not depend on the coefficients $a_j$, \emph{a fortiori} not on the value of $s$. Thus, the eigenvalues are equal to those of the matrix where $a_j =1$ for all $j$, and these are well known to be
\begin{equation}\label{eigs1}
  \lambda_k = 2\cos\frac{k \pi}{N+1}, \quad k = 1,2, \ldots , N.
\end{equation}
Note that the set of eigenvalues is invariant with respect to $x \mapsto -x$. In particular, when $N$ is odd, the middle eigenvalue $\lambda_{(N+1)/2} =0$.
Let the corresponding eigenvectors be denoted $f_k$. These depend on the choice of parameter $s$, but their exact form is of no consequence in what follows. The $N$-tuple of $f_k$ are linearly independent, and each $f_k$ can be identified with a linear function of $z_j$. Next, observe that
$|\vec{\alpha}\rangle: = f_2^{\alpha_2} f_3^{\alpha_3} \ldots f_{p_N}^{\alpha_{p_N}}$, where $|\vec{\alpha}| = n$, furnishes a basis of the $n$-particle space. Moreover, in light of (\ref{Leibnitz}), we have
\begin{equation}\label{LeibD}
  D\, |\vec{\alpha}\rangle = \sum_{k =1}^{N} \alpha_{p_k} \lambda_k \,\, |\vec{\alpha}\rangle.
\end{equation}
In summary, the set of eigenvalues of $D$ in the $n$ particle space is
\begin{equation}\label{eigsn}
  \left\{2 \sum_{k =1}^{N} \alpha_{p_k} \cos\frac{k \pi}{N+1}\,:\,\, |\vec{\alpha}| = n \, \right\}.
\end{equation}
Thus, the eigenvalues are given by an exact formula.

The number operator is also given explicitly by
\begin{equation}\label{sumNp}
  \hat{N} = \sum_{j=1}^{N} \hat{N}_{p_j} = \sum\limits_{j =1}^{N}  z_{p_j}^* \partial_{z_{p_{j}}^*}, \quad \mbox{ and, of course, } \quad  \hat{N}\, |\vec{\alpha}\rangle = |\vec{\alpha}|\, |\vec{\alpha}\rangle.
\end{equation}
Thus, he Hamiltonian of the form
$
  \gamma   \sum_{j=1}^{N} \hat{N}_{p_j} - \tau D
$
has eigenvalues
\begin{equation}\label{eigsn2}
  \lambda_{\vec{\alpha}} = \gamma n + 2 \tau \sum_{k =1}^{N} \alpha_{p_k} \cos\frac{k \pi}{N+1}, \quad  \mbox{ where }\,\, |\vec{\alpha}| = n.
\end{equation}
when restricted to the $n$-particle space. Note that $ \lambda_{\vec{\alpha}}$ corresponds to the state $ |\vec{\alpha}\rangle$.

We are interested in a more general Hamiltonian
\begin{equation}\label{quadratic}
  \delta\left(\gamma   \sum_{j=1}^{N} \hat{N}_{p_j} - \tau D\right)^2 + \gamma   \sum_{j=1}^{N} \hat{N}_{p_j} - \tau D,
\end{equation}
whose $n$-particle eigenvalues are precisely
\begin{equation}\label{eigsn3}
  \delta\,\lambda_{\vec{\alpha}}^2 + \lambda_{\vec{\alpha}}, \quad  \mbox{ where }\,\, |\vec{\alpha}| = n,
\end{equation}
each of which corresponds to the state $ |\vec{\alpha}\rangle$. We discuss the physical properties of Hamiltonian (\ref{quadratic}) in the next section.

\subsection{Physical implications}
The model given in Eq. \eqref{quadratic} comprises a few interesting physical implications that are worth noticing in the present context. First, we outline some elementary facts that are helpful for the physical interpretation of the results.

The eigenvectors $f_k$ associated with the matrix $D$ physically represent linear combinations of the states localized at sites $p_k$, $k=1,2,\ldots,N$. Moreover, since the parameters $\tau$ and $\gamma$ do not depend on the site index, it is clear that the eigenvector $f_{p_1}$ corresponding to the lowest eigenvalue $\lambda_1$, represents a maximally delocalized particle configuration. The degree of particle localization captured by $f_{p_k}$ increases with increasing $k$ in the interval $2\leq k< N_U$, where $N_U=N/2$ or $N_U=(N+1)/2$, depending on whether $N$ is even or odd, respectively. By localization in this context we refer to particles in close vicinity to each, either at the same site or at different sites nearby one another. 

Accordingly, multi-particle states also have localization qualities. Indeed, as mentioned above, the product state $|\vec{\alpha}\rangle$, where $|\vec{\alpha}|=n$, represents an $n$-particle state with $\alpha_{p_1}$ particles in $f_{p_1}$, $\alpha_{p_2}$ particles in $f_{p_2}$, and so forth. Therefore, $f_2^nf_3^0\cdots f_{p_N}^0$ describes the three particle configuration in which all particles occupy the maximally delocalized state $f_{p_1}$. With this understood, we are in a position to analyze and interpret the results from Eq. \eqref{quadratic}.

Considering the non-interacting limit, $\delta=0$, the ground state in a set-up with $N$ sites and $n<N$ particles, is represented by the maximally delocalized state, $|\vec{\alpha}_1\rangle = f_{p_1}^n$, in which all particles occupy the same state. Therefore the state is spread out, and there is at most one particle per site. By the linear dependence on $\tau$, the ground state remains the same for all hopping rates. A pertinent example of the spectrum is plotted in Fig. \ref{fig-SpectrumN5n3} (a) as function of $\tau$, where the mode number on the horizontal axis represents the order of $|\vec{\alpha}_k\rangle$, so that $\lambda_{\vec{\alpha}_1} \leq\lambda_{\vec{\alpha}_2} \leq \ldots \leq \lambda_{\vec{\alpha}_{15}}$. Here, we have limited to the plot to the fifteen lowest modes for the sake of highlighting the properties near the ground state and the lowest excited states. The plot in Fig. \ref{fig-SpectrumN5n3} (a) clearly demonstrates that $|\vec{\alpha}_1\rangle$ remains the ground state for all $\tau$ within the plotting range, as it corresponds to the lowest eigenvalue.

\begin{figure}[t]
\begin{center}
\includegraphics[width=\textwidth]{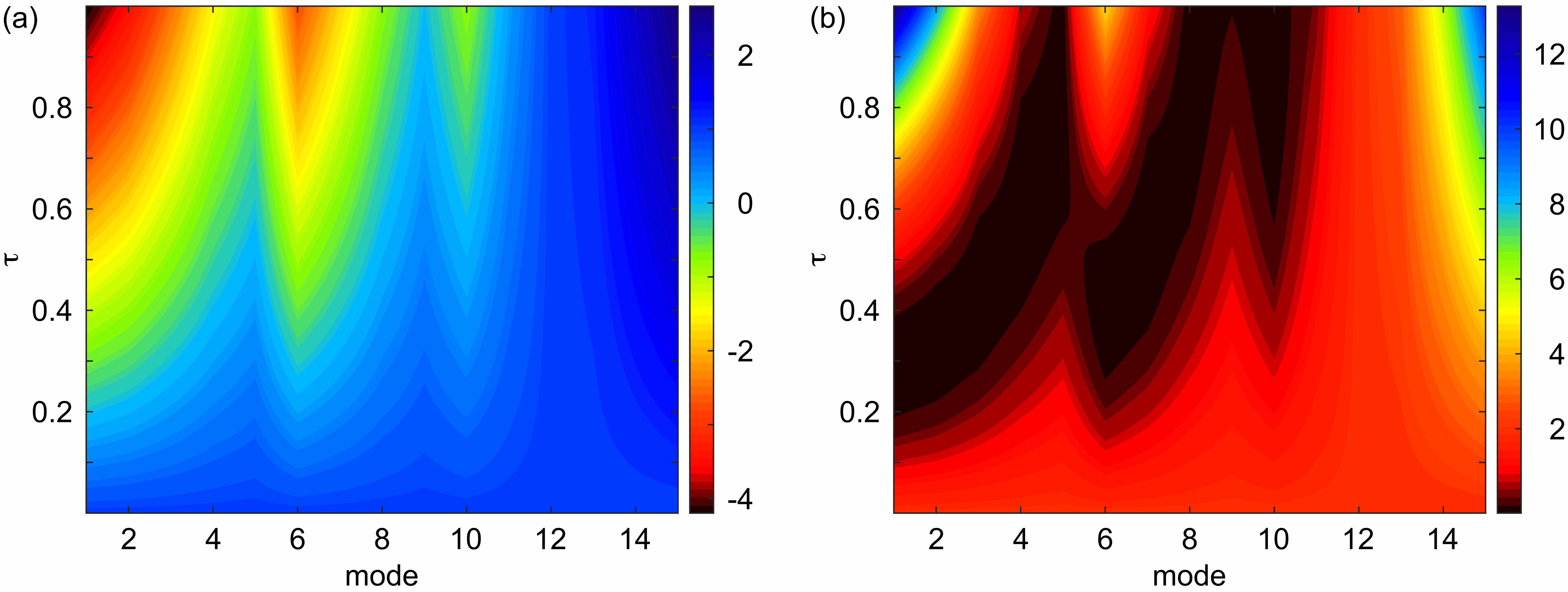}
\end{center}
\caption{Spectra of the model in Eq. \eqref{quadratic} for the modes $|\vec{\alpha}_k\rangle$  with first fifteen lowest energies $\lambda_{\vec{\alpha}_1} \leq\lambda_{\vec{\alpha}_2} \leq \ldots \leq \lambda_{\vec{\alpha}_{15}}$, for (a) $\delta=0$ and (b) $\delta=1$, as function of the hopping rate $\tau$. Here, $N=5$, $n=3$, and $\gamma=1$.}
\label{fig-SpectrumN5n3}
\end{figure}

By contrast, the ground state in the interacting model does not remain the same for all $\tau$. This can be seen in Fig. \ref{fig-SpectrumN5n3} (b), in which the spectrum is recalculated with $\delta=1$. The eigenvalues of all modes between $|\vec\alpha_1\rangle$ and $|\vec\alpha_{11}\rangle$ are strikingly non-monotonic. Moreover, for a range of hopping rates $0.4<\tau<1$, the maximally delocalized state, $|\vec\alpha_1\rangle$, is not the ground state; instead, the ground state is formed by two particles with the energy $\lambda_1$ and one with another energy, e.g., $\lambda_2$, $\lambda_3$, or $\lambda_4$. The interacting system, hence, undergoes configurational transitions for a range of hopping rates $\tau$ which may have correspondences as phase transitions in the thermodynamical limit, $N\rightarrow\infty$. This is, nevertheless, beyond the scope of the present article.

Despite the repulsive overall character of the interaction, the transition of the ground state from the maximally delocalized state to configurations which are less delocalized is conspicuous. The transition between the states can be understood as the effect of an anisotropic interaction which drives the solution away from the maximally delocalized state for a range of hopping rates. This may be compared to anisotropy reflected by the Ising model in comparison with the Heisenberg model for spin interactions. While the latter is purely isotropic, it provides a reasonable description of either ferromagnetism or anti-ferromagnetism, however, without any preferred spatial direction. The anisotropic contribution from the Ising model introduces such preference which thereby also restricts the solution space of the combined models.

%%%%%%%%%%%%%%%%%%%%%%%%%%%%%%
%%%%%%%%%%%%%%%%%%%%%%%%%%%%%%
%%%%%%%%%%%%%%%%%%%%%%%%%%%%%%
%%%%%%%%%%%%%%%%%%%%%%%%%%%%%%

\section*{Acknowledgements} The authors acknowledge partial support via the Global Ambassador Program, University of Saskatchewan 2022-23 and 2023-24.

\end{document}